\documentclass[
aps,prx,
a4paper,
preprintnumbers,
floatfix,
reprint,
longbibliography,
nofootinbib,
]{revtex4-2}

\usepackage[pass]{geometry}
\usepackage[utf8]{inputenc}
\usepackage[T1]{fontenc}
\usepackage{amsmath,amssymb,amsthm}
\usepackage{mathtools}
\usepackage{braket}
\usepackage{xcolor}
\usepackage{graphicx}
\usepackage{booktabs}
\usepackage[most]{tcolorbox}
\usepackage{enumerate}
\usepackage{setspace}
\usepackage{xr}
\usepackage{subfigure}

\definecolor{pgreen}{RGB}{84, 129, 102}
\definecolor{porange}{RGB}{199, 103, 42}
\usepackage[
colorlinks,
citecolor=pgreen,
linkcolor=porange,
urlcolor=pgreen
]{hyperref}

\newenvironment{mybox}[1]{%
\begin{center}%
\begin{tcolorbox}[enhanced,title={#1},width=.99\linewidth,breakable]%
}{\end{tcolorbox}\end{center}}

\newcommand*\eye{{\normalfont\openone}}
\newcommand*\prob[1]{\mathbb{P}\left[{#1}\right]}
\newcommand*\expv[1]{\mathbb{E}\left[{#1}\right]}
\newcommand*\dyad[1]{{\ket{#1}\!\!\bra{#1}}}
\newcommand*\dist\thicksim
\newcommand*\dd{\mathrm d}
\renewcommand*\vec\boldsymbol
\renewcommand*\tilde[1]{\widetilde{#1}}

\renewcommand*\paragraph[1]{\section{#1}}

\let\Re\undefined
\let\Im\undefined
\DeclareMathOperator\Re{Re}
\DeclareMathOperator\Im{Im}
\DeclareMathOperator\diag{diag}
\DeclareMathOperator\tr{tr}
\DeclarePairedDelimiter\abs\lvert\rvert
\DeclarePairedDelimiter\norm\lVert\rVert

\newtheorem{theorem}{Theorem}
\newtheorem{corollary}[theorem]{Corollary}
\newtheorem{lemma}[theorem]{Lemma}
\theoremstyle{remark}
\newtheorem{remark}{Remark}[section]

\begin{document}
\let\oldaddcontentsline\addcontentsline
\renewcommand{\addcontentsline}[3]{}

\title{User-friendly confidence regions for quantum state tomography}

\author{Carlos de Gois}
\email{carlos.bgois@uni-siegen.de}
\affiliation{Naturwissenschaftlich-Technische Fakult\"{a}t, Universit\"{a}t Siegen, Walter-Flex-Stra\ss e 3, 57068 Siegen, Germany}

\author{Matthias Kleinmann}
\email{matthias.kleinmann@uni-siegen.de}
\affiliation{Naturwissenschaftlich-Technische Fakult\"{a}t, Universit\"{a}t Siegen, Walter-Flex-Stra\ss e 3, 57068 Siegen, Germany}

\begin{abstract}
    Quantum state tomography is the standard technique for reconstructing a quantum state from experimental data. In the regime of finite statistics, experimental data cannot give perfect information about the quantum state. A common way to express this limited knowledge is by providing confidence regions in the state space. Though other confidence regions were previously proposed, they are either too wasteful to be of practical interest, cannot easily be applied to general measurement schemes, or are too difficult to report. Here we construct confidence regions that solve these issues, as they have an asymptotically optimal sample cost and good performance for realistic parameters, are applicable to any measurement scheme, and can be described by an ellipsoid in the space of Hermitian operators.
    Our construction relies on a vector Bernstein inequality and bounds with high probability the Hilbert--Schmidt norm error of sums of multinomial samples transformed by linear maps.
\end{abstract}

\maketitle

\paragraph{Introduction}%
    Quantum information processing relies on the precise control of increasingly higher-dimensional quantum systems. Ultimately, this ability will enable the development of quantum communication systems, quantum metrology and quantum computing. However, to ensure the correct functioning of any quantum information processing device, one must be able to reliably characterise the states being prepared. Quantum state tomography is the cornerstone technique to obtain a full description of a quantum state based solely on experimental data.

    A central practical difficulty in quantum state tomography is that it fundamentally requires measuring an exponential number of preparations of the same state, with respect to the number of qudits. There are two main approaches to alleviate this requirement. One is to impose prior constraints on the state, such as it being pure, permutationally symmetric, etc.\ \cite{kohout2012robust,kohout2010optimal,shang2013optimal,gross2010quantum,moroder2012permutationally,cramer2010efficient}, and the other is to refrain from having a complete description of $\rho$, and settle for only estimating, for example, the two-body marginals of the state \cite{bonet2020ot,cotler2020ot} or the mean values of a collection of local observables \cite{aaronson2018shadow,huang2020predicting}. For all other cases, quantum state tomography remains a standard technique.

   \begin{figure}
        \centering
        \includegraphics[width=.55\columnwidth]{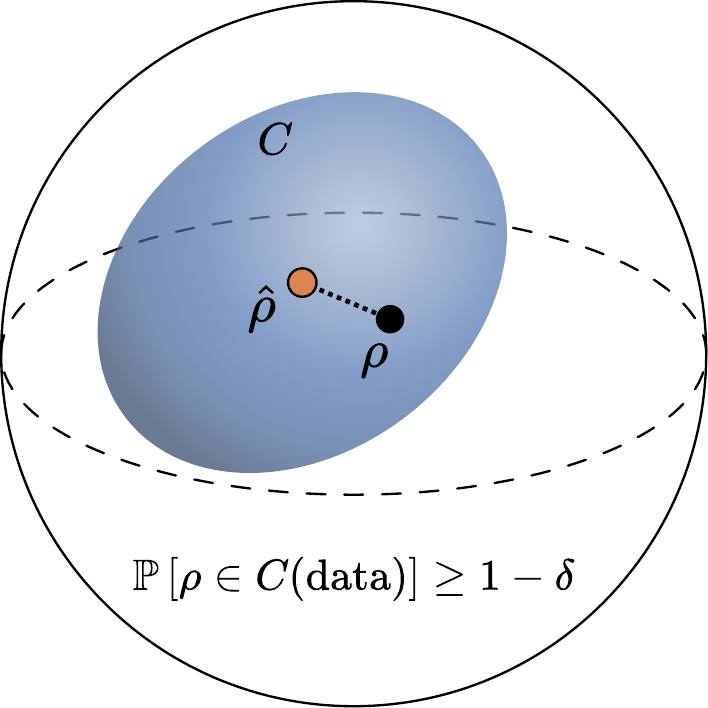}
        \caption{The goal of quantum state tomography is to reliably reconstruct a quantum state based on measurement data. In the finite data regime, any estimate $\hat{\rho}$ will almost surely diverge from the true state $\rho$ to a certain extent. This divergence can be quantified by means of confidence regions. A confidence region for quantum state tomography is a region (the blue ellipsoid) in some space (for example, the state space, or the space of Hermitian operators) that contains the true state $\rho$ with high probability. Thus, it provides quantitative guarantees on how reliable the reconstruction $\hat\rho$ is. Our main result is the construction of confidence regions described by ellipsoids in the space of Hermitian operators. The regions are efficient in the number of samples, of general applicability, and simple to report.}
        \label{fig:intro}
    \end{figure}
    
    As an inherently statistical procedure, tomography cannot exactly characterise a quantum state from finite data. This is unrelated to errors or imperfections in the experimental setup, but rather, it happens for the same reason why one cannot prove a dice is perfectly unbiased with any finite number of trials. Therefore, as much as any other experimental data, a state $\hat{\rho}$ estimated in an experiment must be presented with a guarantee of precision. From a frequentist point of view, these guarantees are typically provided in terms of confidence regions.  
    In a tomography setting, confidence regions are conveniently specified as a region $C$ that contains the true state $\rho$ with a specified confidence level $1 - \delta$, see Fig.~\ref{fig:intro}.
    Then,
    \begin{equation}
        \prob{\rho \in C(\text{data})} \geq 1 - \delta,
    \label{eq:confidence-regions-general}
    \end{equation}
    where the data can be, for instance, the outcome frequencies from a set of measurements.
    The confidence level $1-\delta$ is usually close to unity, for example a ``$3\sigma$'' confidence level corresponds to $1-\delta \approx 0.9973$ and, generally, confidence levels become higher with an increasing number of samples.
    
    To be useful in practice, a confidence region should satisfy three conditions.
    Since sampling and measuring quantum states is expensive, (1) the region should be reasonably tight. Because experimental setups may use specifically-tailored measurement schemes, it is also important that (2) the confidence region can be constructed easily, irrespective of which measurement scheme is used. Moreover, (3) it should be possible to report the region associated to the experimental data in a simple and informative way.

    Confidence regions for unconditional, state-agnostic tomography were already proposed in the literature \cite{christandl2012reliable,sugiyama2013precision,faist2016practical,wang2019confidence,guta2020fast}. Their performances in simulated experiments were recently compared quantitatively \cite{almeida2023comparison} and two outstanding regions have been identified.
    One of them is a polyhedral region \cite{wang2019confidence} which is simple to adapt for different measurement schemes, satisfying criterion (2), but that does not have an explicit tightness guaranteed. As pointed out by the authors, reporting the region is not straightforward, because it requires listing a large number of facets of the polyhedron.
    The other region is based on the spectral norm \cite{guta2020fast}. It is close to optimal in the asymptotic limit and with respect to this norm, so (1) holds, but it is generally very difficult to adapt to general measurement schemes. For the measurement schemes where it is applicable, it is easy to report the region, so criterion (3) is also fulfilled.

    The regions that we present satisfy all three criteria. They have the shape of a Hilbert--Schmidt norm sphere or ellipsoid in the space of Hermitian operators, which can be described through their radius or semiaxes. On top of their simple geometry, the regions are applicable to any measurement scheme, and can be integrated into convex optimisation procedures. This last property is shared with the previous two regions, and is crucial to determine whether the true state is outside of a given convex set, and also to provide confidence intervals for functions of the state.

\paragraph{State tomography and confidence regions}%
\label{sec:state-tomography-and-crs}
    In a tomography procedure, a large number $N$ of systems in the same state $\rho$ is prepared and independently measured. If the measurements have enough structure to provide full information on the state, then it is possible to obtain an estimate $\hat{\rho}$ that converges to $\rho$ with increasing $N$.  
    To formalise this, we consider different measurement settings, labelled by $s$. In quantum theory, a measurement is described by operators $E_{a \vert s}$, each associated with a possible outcome $a$ from the measurement setting $s$. The operators are positive semidefinite, $E_{a \vert s} \succcurlyeq 0$, and sum to identity, $\sum_a E_{a \vert s} = \eye$. Measuring setting $s$ on a quantum state $\rho$ returns the outcome $a$ with probability $p_{a \vert s} = \tr(E_{a \vert s} \rho)$. For convenience, we describe the outcome probabilities of setting $s$ by the vector $\vec{p}_s = ( p_{a \vert s})_a$ and we also combine all such vectors into a single vector, $\vec p=\bigoplus_s\vec p_s$. Now suppose we have measured $n_s$ of the total $N$ state preparations with setting $s$ and counted $c_{a \vert s}$ occurrences of outcome $a$. Then the frequency of this outcome is $f_{a \vert s} = c_{a \vert s}/n_s$ and $f_{a \vert s}$ converges to $p_{a \vert s}$ with increasing $n_s$. Doing the same for each setting, we obtain the collection of frequencies $\vec{f}_s$ and the vector of all frequencies $\vec f$. 

    In full quantum state tomography, the measurements used to obtain $\vec{f}$ have to be tomographically complete. This means that the total measurement map $M\colon \rho \mapsto \vec{p}$ can be inverted, yielding the point estimate $\hat\rho(\vec f)=M^+\vec f$ for the measured data. The map $M^+$ can be any left-inverse of $M$, that is, any map where $M^+M\rho=\rho$ holds for all $\rho$. A common choice is the pseudoinverse, which corresponds to the free least-squares estimator; see Ref.~\cite{acharya2019comparative} for a comparison between several common point estimators.
    Important examples of measurement schemes are the Pauli-bases measurement, global Pauli observables, and structured measurements. For Pauli-bases, one considers $q$ qubits, where each is independently measured with one of the Pauli observables $X$, $Y$, or $Z$ (App.~\ref{ap:pauli-tomography}) while for Pauli observables (App.~\ref{ap:pauli-obs2}, one globally measures tensor products of $\eye$, $X$, $Y$, and $Z$. Structured measurements \cite{guta2020fast} constitute a class of highly symmetric global measurements (App.~\ref{ap:sic-v3}. They include the symmetric informationally complete measurements \cite{renes2004sic}, maximal sets of mutually unbiased bases \cite{durt2010mubs}, and the uniform measurement.
    
    The estimator can be used to define norm-based confidence regions, that is, confidence regions based on statements of the form
    \begin{equation}
        \prob{\norm*{ \rho - \hat{\rho}(\vec f) }_\star \geq \epsilon} \leq \delta,
    \label{eq:confidence-regions-norm}
    \end{equation}
    where $\star$ is a chosen norm, $\epsilon$ stands for the radius of the region under the given norm, and $1 - \delta$ is the confidence level. In general, $\delta$ depends on several parameters, such as the number of samples, the state dimension, the chosen norm and $\epsilon$. An inequality of this type implies that
    $\norm{\rho-\hat\rho(\vec f)}_\star < \epsilon$ is true with probability of at least $1 - \delta$, thus providing rigorous guarantees that the estimate is close to the true state.

\paragraph{Constructing the confidence regions}%
    From our previous discussion, it follows that $\hat{\rho}(\vec f)$ can be interpreted as a sum of matrix-valued or vector-valued random variables, and that Eq.~\eqref{eq:confidence-regions-norm} measures how much it may deviate from the true state. Concentration bounds, such as matrix- or vector-Bernstein inequalities can provide rigorous guarantees \cite{ahlswede2002strong,gross2011recovering,minsker2017some,tropp2012user,tropp2015introduction,vershynin2018high,wainwright2019high}. In particular, matrix bounds have been previously applied to quantum state tomography \cite{gross2010quantum,guta2020fast}. Following similar ideas, we derive two confidence regions for quantum state tomography. However, instead of starting from a matrix bound, we use a vector Bernstein inequality as our cornerstone:
    Given independent, zero-mean, vector-valued random variables $X_1, \ldots, X_N$, it holds that
    \begin{equation}
        \prob{ \norm*{ \frac{1}{N} \sum_i X_i }_2 \geq \epsilon\sigma}
        \leq 8 \exp \left[ -\frac{N\epsilon^2}2\,w \left( \epsilon \frac L\sigma \right) \right].
    \label{eq:vector-bernstein-main}
    \end{equation}
    Here, $\norm{\vec x}_2$ is the Euclidean norm, $w(x) = 3/(3+x)$, the parameter $\sigma^2$ bounds the average of the variances of $\norm{X_i}_2$, and $L$ is a bound on all $\norm{X_i}_2$; for details, see Appendix \ref{ap:vector-bernstein}.

    \begin{table}
    \begin{tabular}{lcc} \toprule
        {} & {$\sigma_A$} & {$\sigma'_B$} \\ \midrule
        Pauli-bases & $d^{1.1610}$ & $d^{1.2925}$ \\
        Pauli observables & $d^{3/2}$ & $d^{3/2}$ \\
        structured measurements & $\sqrt{5/4}\, d$ & $\sqrt{3/2}\, d$ \\ \bottomrule
    \end{tabular}
    \caption{
    Expressions for the radius $\epsilon \sigma_A$ of confidence region $C_A$ and the largest semiaxis $\epsilon\sigma_B'$ of region $C_B$ depending on the total dimension $d$ of the system. For the Pauli-bases, each qubit is measured with all Pauli matrices $X$, $Y$, $Z$, while for Pauli observables, all combination of Pauli matrices, including $\eye$, are measured globally. Structured measurements \cite{guta2020fast} are a general class of qudit measurements which includes, for example, maximal sets of mutually unbiased bases and symmetric informationally complete measurements. Tighter bounds and all semiaxes are provided in Appendix \ref{ap:explicit-construction}.}
    \label{tab:sigmas}
    \end{table}
    
    To make Eq.~\eqref{eq:vector-bernstein-main} suitable for quantum state tomography, we consider data coming from multinomial distributions $\mathcal D(\vec{p}_s; n_s)$ with parameters $\vec p_s$ and a total of $n_s$ samples; the index $s=1,\dotsc,m$ allows us to combine data from different distributions. We write $\vec f_s$ for the frequencies, hence $\vec f_s\dist\mathcal D(\vec{p}_s; n_s) / n_s$.
    So that we can later model the state estimator, it is convenient to add linear maps $G_s$ and consider the sum of random vectors $\frac 1N\sum_s G_s (\vec{p}_s - \vec{f}_s)$. This can be substituted into Eq.~\eqref{eq:vector-bernstein-main}, but to apply it we are still bound to determine $\sigma$ and $L$.
    We argue that not much is lost by choosing $\sigma^2 = \sum_s \gamma_s^2/q_s$ and $L = 2 \max_s(\gamma_s/q_s)$, where $q_s=n_s/N$ and $\gamma_s \equiv \norm{G_s}_{2,\infty}$ stands for the maximum of the Euclidean norms of the column vectors of $G_s$ (App.~\ref{ap:bernstein-multinomial}).
    The resulting inequality holds for general multinomial distributions and linear maps $G_s$.
    For tomography experiments, it is convenient to combine all measurement settings into a single generalised measurement, that is, $m=1$, $\sigma=\norm{G}_{2,\infty}$, and $L=2\sigma$ (see Sec.~\ref{sec:usage} and App.~\ref{ap:construct-general}).

    Our first confidence region ensues when setting $G = M^+$. Since $M^+ (\vec p - \vec f) = \rho - \hat\rho(\vec f)$, we obtain the confidence region
    \begin{equation}\label{eq:region-ca}
        C_A(\epsilon; \vec f)
        = \set{ \rho \mid \norm*{\rho-\hat\rho(\vec f)}_\mathrm{HS} \leq \epsilon \sigma_A}.
    \end{equation}
    Here, $\norm{A}_\mathrm{HS}=\sqrt{\tr(A^\dag A)}$ is the Hilbert--Schmidt norm and $\epsilon$ can be determined from the chosen confidence level $1 - \delta$ through
    \begin{equation}\label{eq:epsilon-main}
        \epsilon = 3 \sqrt{u} (\sqrt u+\sqrt{u+1})
    \end{equation}
    and $u = 2 \log(8/\delta)/9 N$.
    Table \ref{tab:sigmas} shows values of $\sigma_A$ for common measurements.
    Geometrically, this region describes a sphere in the space of Hermitian operators, which is centred at $\hat\rho(\vec f)$ and has radius $\epsilon\sigma_A$.
    For the second region, we use $G = MM^+$. This yields
    \begin{equation}\label{eq:region-cb}
        C_B(\epsilon;\vec f)
        = \set{\rho \mid \norm*{\rho-\hat\rho(\vec f)}_M\le \epsilon \sigma_B },
    \end{equation}
    where $\norm{A}_M \equiv \norm{M A}_2$, and $\epsilon$ is again given by Eq.~\eqref{eq:epsilon-main}.
    Region $C_B$ is an ellipsoid in the space of Hermitian operators, and its semiaxes can be determined from $\sigma_B$ and the diagonalisation of $M^\dag M$. The largest semiaxis $\epsilon\sigma_B'$ for common measurements are provided via Table \ref{tab:sigmas}.

\paragraph{Usage}%
\label{sec:usage}

    To report either of these confidence regions in a tomography experiment, it is only necessary to specify the estimated state $\hat{\rho}$ and the radius or the semiaxes of the region. 
    For clarity, we first describe how these can be obtained in the case of an arbitrary generalised measurement. Let us assume that the measurement is described by the effects $(E_1, E_2, \ldots)$ and that it was performed $N$ times, yielding the frequencies $\vec{f} = (f_1, f_2, \ldots)$. For the chosen confidence level $1 - \delta$, the value of $\epsilon$ can be computed as in Eq.~\eqref{eq:epsilon-main}. To determine $\sigma$, it is convenient to first choose a vectorisation $\vec{v}(A) \in \mathbb{R}^{d^2}$ of the $d$-dimensional Hermitian operators with the property that $\vec{v}(A) \cdot \vec{v}(B) = \tr(AB)$. A common choice is to let $\vec{v}(A)$ concatenate the rows of $A$. Then:
    \begin{enumerate}
        \item Construct a matrix $M$ where row $i$ is $\vec{v}(E_i)$.
        \item Compute the pseudoinverse $M^{+}$ of $M$ (this can be done numerically, symbollically or analytically).
        \item Determine the state estimate via $\hat{\rho} = \vec{v}^{-1}(M^+ \vec{f})$.
        \item For region $C_A$, compute $\sigma_A$ as the maximal Euclidean norm of the columns of $M^+$.
        \item Otherwise, for region $C_B$, compute $\sigma_B$ as the maximal Euclidean norm of the columns of $MM^+$. The semiaxes of the ellipsoid describing the region are given by $\epsilon A_i$, where to compute $A_i$ we perform a spectral decomposition of $M^\intercal M$ yielding the eigenvectors $\vec{x}_i$ and eigenvalues $\xi_i$, then set $A_i = \sigma_B \vec{v}^{-1}(\vec{x}_i) / \sqrt{\xi_i}$.
    \end{enumerate}
    The case for multiple measurement settings is discussed in App.~\ref{ap:construct-general}.
    
    In situations where sampling is particularly expensive, it is possible to optimise the confidence regions by various means. For example, since ${\sigma}_{A/B}$ depends on the choice of the estimator $\hat \rho$, one can consider to optimise the left-inverse of $M$ (App.~\ref{ap:optim-leftinverse}). Complementarily, one can optimise how the samples are distributed among the different measurement settings (App.~\ref{ap:sampling-strategies}).
    
    Our confidence regions can readily be used to compute bounds on functions $f(\rho)$. In particular, since the states contained in these regions can be described as a convex cone constraint, convex functions $f$ can be minimised over $\rho \in C_{A/B}(\hat{\rho})$ using standard optimisation techniques. Explicitly, one computes $\min_\rho\set{ f(\rho) | \rho\in C_{A/B} \text{ and } \rho\ge 0}$. Functions suitable for this method include the expectation value of observables, the fidelity to a pure state, the best-case fidelity to a mixed state, the negativity of entanglement, etc. The same idea can be used to perform a hypothesis test for $\rho$ to be outside a convex  set $\mathcal S$, by verifying that the confidence region does not contain any state in $\mathcal S$, for example for a robust certification of genuine multipartite entanglement (App.~\ref{ap:gme}).

    Even though we have presented the regions in the context of full state tomography, the same ideas apply to a variety of other problems, such as tomography with an incomplete measurement basis, and reduced marginals tomography \cite{bonet2020ot,cotler2020ot,perez2020pairwise}. This holds because our only requirement is that the estimated quantity can be computed from some probabilities $p_{a|s}=\tr(E_{a|s}\rho)$ by means of a linear map.
    
\paragraph{Performance analysis}%
    \begin{figure}
        \centering
        \includegraphics[width=0.78\columnwidth]{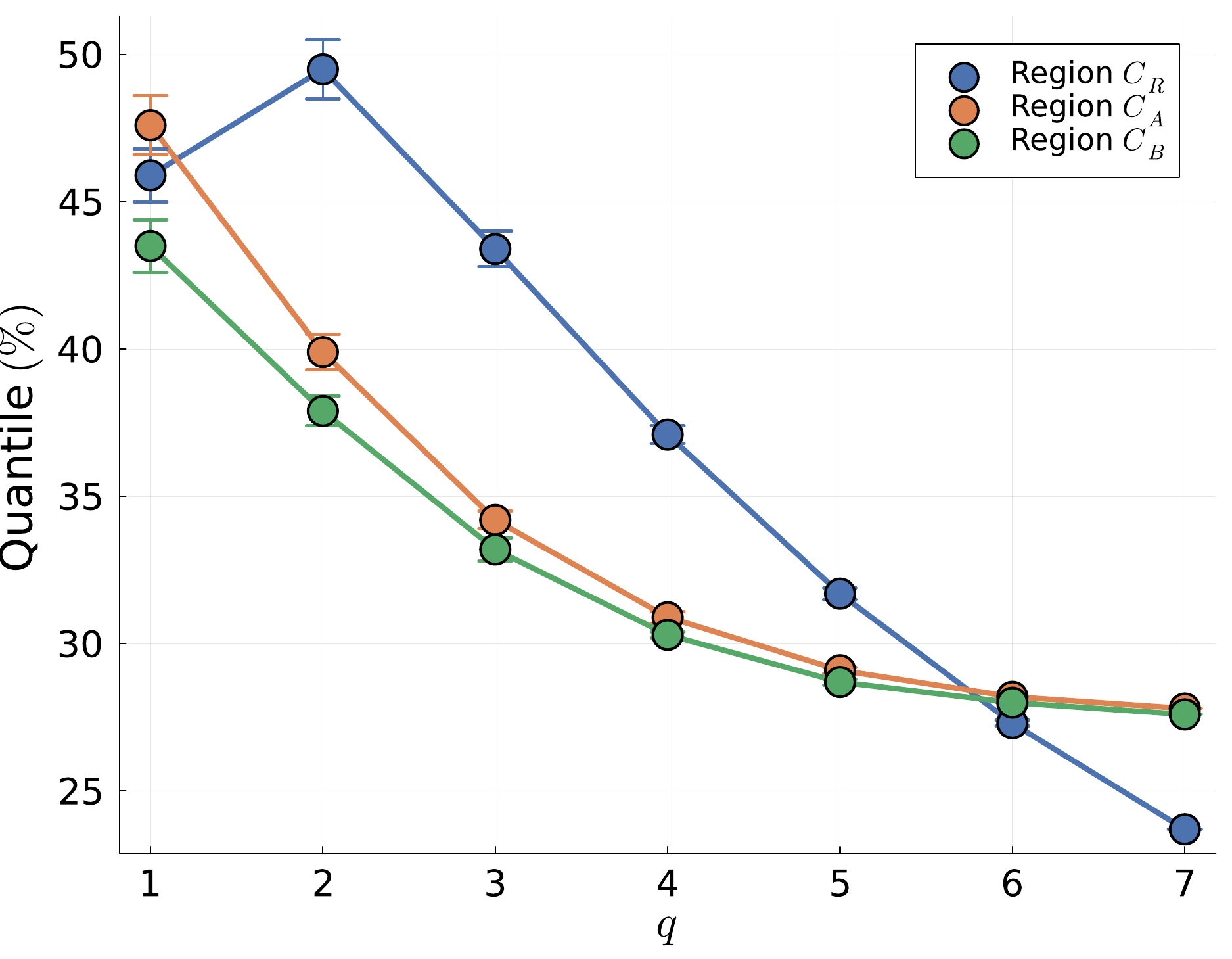}
        \caption{Tightness of confidence regions $C_R$, $C_A$, and $C_B$ in terms of the $1-\delta$ quantiles for the ratio $r=\norm{\rho-\hat\rho}_\star/(\epsilon\sigma)$. The distance $\norm{\rho - \hat{\rho}}_\star$ between the estimated state $\hat{\rho}$ and the fiducial state $\rho$ is compared to the size of the confidence region $\epsilon\sigma$.
        The $1-\delta$ quantile is a measure of the tightness of the confidence region, and would be $100\%$ for a tight region.
        Here, the distribution of the ratios $r$ is sampled over $10\,000$ reconstructions of $\rho$, the confidence level is $1 - \delta=0.99$, the number of simulated state preparations for each reconstruction is $N=60\,000$, and the Pauli-basis measurement scheme is simulated. The procedure was repeated for $50$ different initial Haar-random pure states and the largest difference in the ratios is presented as the uncertainty.}
        \label{fig:ratios}
    \end{figure}%
    One common way of evaluating the performance of a confidence region is by its asymptotic sample cost. From the previous discussion, we have that $\hat{\rho}$ is $\epsilon$-close to $\rho$ in the Hilbert--Schmidt norm and with confidence level $1 - \delta$ when the number of samples is
    \begin{equation}
        N = O \left[ \left(\frac{{\sigma}_{A}}{\epsilon} \right)^2 \log\left( \frac{1}{\delta}\right) \right].
    \end{equation} 
    For instance, using a structured measurement and a fixed confidence level, we get $N = O(d^2 / \epsilon^2)$, matching the best known sample cost for the Hilbert--Schmidt norm \cite{flammia2023quantum, kueng2017low}. Using the trace norm, $\norm{ X }_1=\tr|X|$ and $\norm{X}_1 \leq \sqrt{d} \norm{ X }_\mathrm{HS}$, it also implies that $\rho$ can be $\epsilon$-approximated in the trace norm with $N = O(d^3/\epsilon^2)$ samples. This agrees with the lower bound on the sample cost for nonadaptive, single-copy measurements \cite{haah2016sampleoptimal,lowe2022lower}, therefore region $C_A$ is asymptotically tight for structured measurements. This extends the results from Ref.~\cite{guta2020fast}, where optimality was only established for the uniform measurement. For the other measurement schemes considered therein, our results improve over their trace-norm sample costs by a factor or order $\log(d)$.
    
    \begin{figure*}
        \centering
        \subfigure[$\mathcal{D}_{1/2} (\ket{0\dotsc0}) $ and $\mathcal{D}_{1/2}(\ket{1\dotsc1})$]{\includegraphics[width=.3\linewidth]{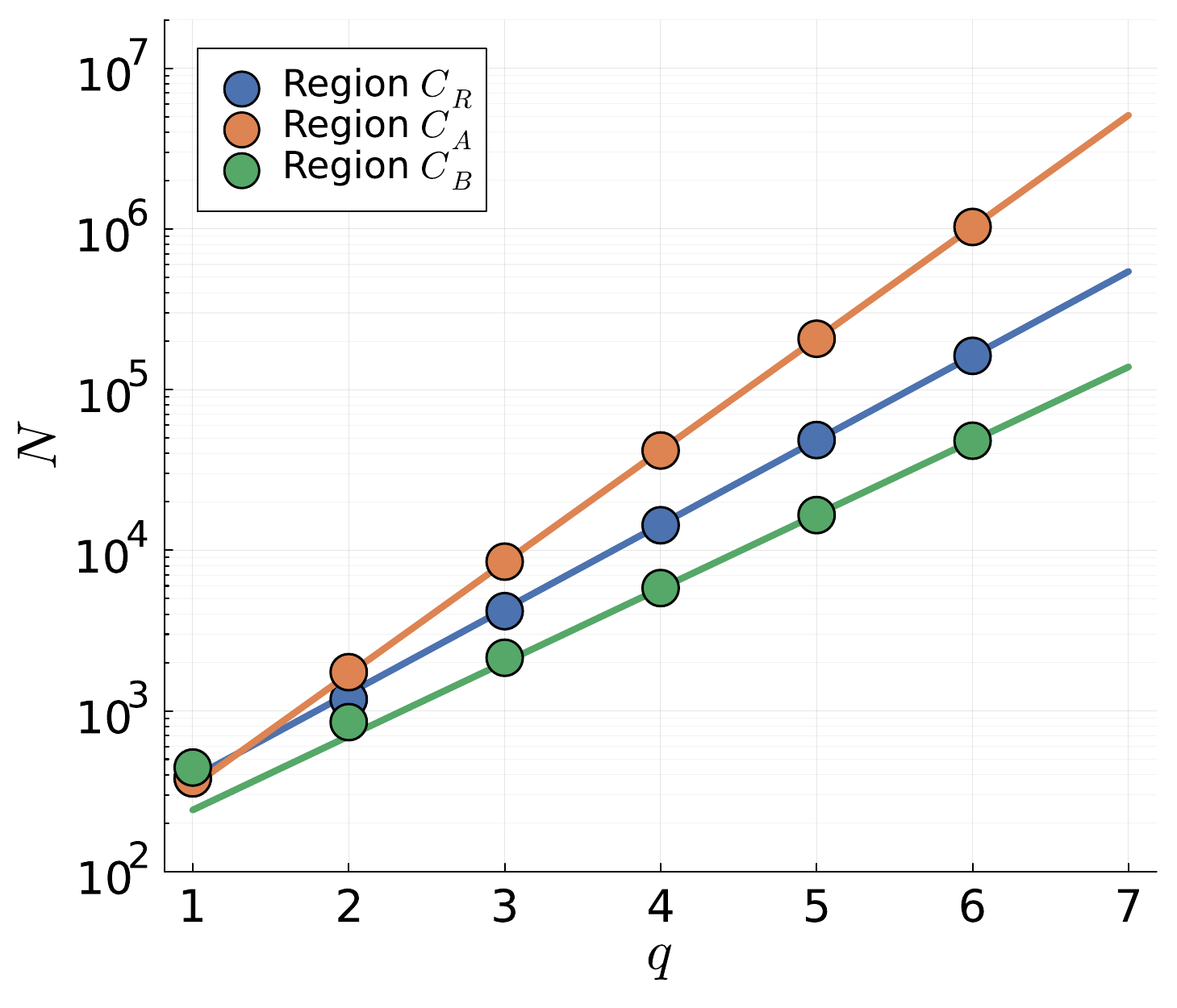}}
        \subfigure[$\ket{\text{GHZ}}$ and $\ket{\text{W}}$]{\includegraphics[width=.3\linewidth]{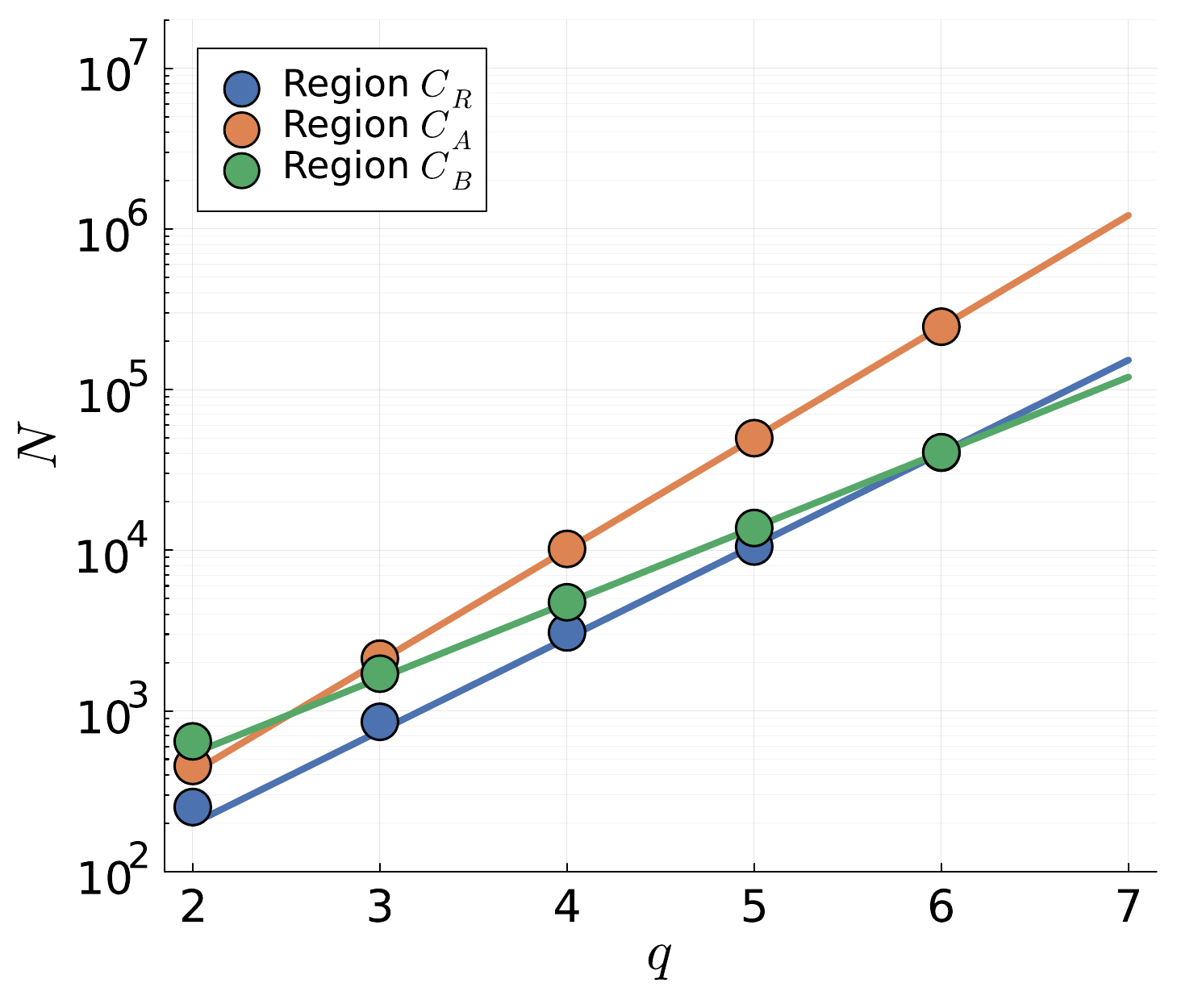}}
        \subfigure[$\ket{\text{GHZ}}$ and $\eye/2^{q}$]{\includegraphics[width=.3\linewidth]{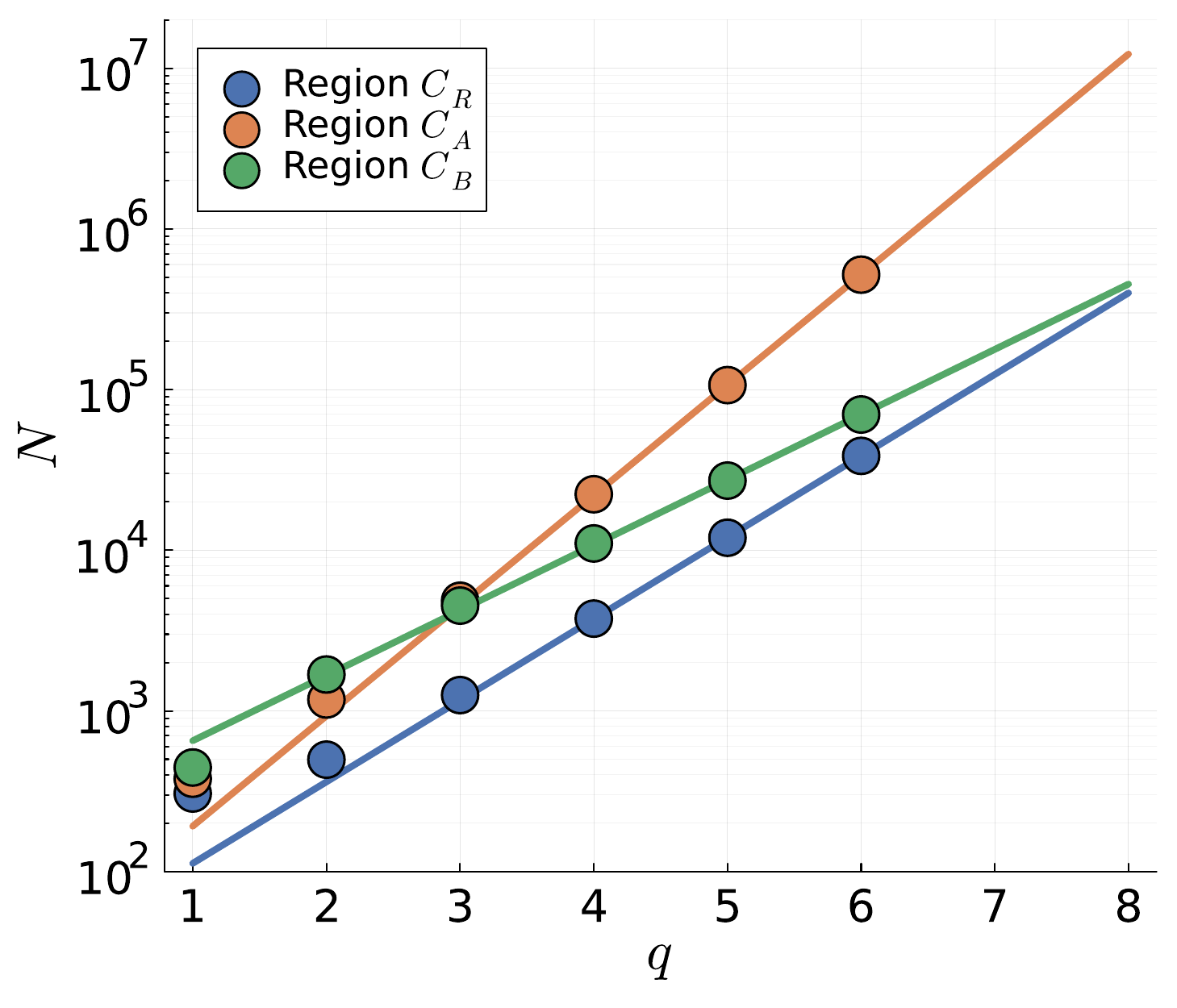}}
        \caption{Distinguishing power of the confidence regions $C_R$, $C_A$ and $C_B$.
        The plots show the number of samples $N$ for which the confidence regions of a pair of states, at a confidence level $1-\delta=0.9$, no longer overlap.
        We consider different pairs for $q=1,\dotsc,6$ qubits, namely,
        (a) $\mathcal{D}_{1/2} (\ket{0\dotsc 0})$ and $\mathcal{D}_{1/2}(\ket{1\dotsc1})$ with $\mathcal{D}_t(\ket\psi) = t\dyad{\psi} + (1-t)\eye/2^q$,
        (b) Greenberger--Horne--Zeilinger (GHZ) state and W state with $\ket{\mathrm{GHZ}}=(\ket{0\dotsc0}+\ket{1\dotsc1})/\sqrt 2$ and $\ket{\mathrm{W}}=(\ket{00\dotsc01}+\dotsm+\ket{10\dotsc00})/\sqrt{q}$,
        (c) GHZ state and maximally mixed state.
        No notable difference was observed when local random rotations were applied simultaneously to both of the fiducial states.
        We simulate tomography experiments using the Pauli-bases measurement and determine $N$ over 1024 estimation procedures such that in $(50 \pm 2)\%$ of the cases the confidence regions do not overlap.
        To each dataset we have fitted an exponential model via least-squares minimisation.}
        \label{fig:plots-intersection-frequencies}
    \end{figure*}

    The issue with the asymptotic sample complexity is that it does not predict how well the region performs in practice, since even a method with worse asymptotic complexity can be better in realistic settings.
    To address this type of concern, two region-agnostic tests were proposed which allow to asses the performance of a confidence region in simulated tomography experiments \cite{almeida2023comparison}. With respect to these tests, it was found that the region constructed in Ref.~\cite{guta2020fast} has a competitive performance, in particular for the Pauli-bases measurement.
    We use an amended form of this region (see App.~\ref{ap:improved-guta}) as a reference region $C_R$ and compare it to regions $C_A$ and $C_B$. It can be written as
    \begin{equation}\label{eq:guta-region-main}
        C_R(\epsilon; \vec f) = \set{\rho \mid \norm{\rho-\hat\rho(\vec f)}_\infty\le\epsilon\sigma_R},
    \end{equation}
    where $\norm{A}_\infty$ is the spectral norm, $\epsilon = 6\sqrt{u}(\sqrt u+\sqrt{u+1}) / \eta$, $u = \eta^2\log(2d/\delta)/ 18 N$, and $d$ is the dimension of the quantum system.  The constants $\sigma_R$ and $\eta$ depend on the chosen measurements and have been computed for a few measurement schemes \cite{guta2020fast} (see also App.~\ref{ap:improved-guta}). For the Pauli-bases measurement on $q$ qubits, they are $\sigma_R = 3^{q/2}$ and $\eta= 2\, (4/3)^{q/2}$.

    Let us first compare these regions with respect to the ratios test, which quantifies the tightness of a confidence region \cite{almeida2023comparison}. The first step is to simulate a tomography experiment by sampling from the multinomial distribution associated with some state $\rho$ and a set of tomographic measurements. From this sample, an estimate $\hat{\rho}$ is constructed. One defines $r=\norm{\rho - \hat{\rho}}_\star /(\epsilon \sigma)$ as the ratio between the distance of the estimator to the true state, and the radius of the confidence region.
    Repeating the simulated tomography will lead to a slightly different $\hat{\rho}$, which in turn leads to a different value of $r$. After many repetitions, one obtains a distribution of ratios. From this distribution of ratios, the $1 - \delta$ quantile is computed. By construction, this quantile tends to one for a tight confidence region, and is lower than one for a less tight region.
    
    We implemented this test for $C_R$, $C_A$, and $C_B$ using the Pauli-bases measurement for up to 7 qubits. The results are shown in Fig.~\ref{fig:ratios}. It can be seen that, up to $5$ qubits, region $C_R$ performs better, but it is overtaken by regions $C_A$ and $C_B$ for large systems. Also, region $C_A$ performs slightly better than $C_B$ for a small number of qubits but quickly turns equivalent.

    Though this test can tell us how tight a confidence region is, it is not particularly informative regarding the number of samples needed to confidently reconstruct a quantum state from measurement data. This can be better addressed by the distinguishing power of the confidence regions \cite{almeida2023comparison}.
    For this test, one chooses two fiducial states $\rho_1$ and $\rho_2$ and simulates tomography with $N$ samples for each of the states.
    The respective confidence regions are then constructed and one verifies whether there exists a common quantum state $\rho$ in both confidence regions.
    Clearly, the probability that such a common state exists reduces when increasing the number of samples.
    For comparison purposes, it is convenient to estimate the number of samples $N$ for which this probability is approximately $1/2$.
    
    We performed this procedure for three different pairs of states, as presented in Fig.~\ref{fig:plots-intersection-frequencies}. While no region outperforms the others in all cases, the results indicate that region $C_A$ is typically weaker than $C_R$, while $C_B$ is comparable to region $C_R$ and tends to outperform it for larger systems. Lastly, we note that the ellipsoidal shape of region $C_B$ is reminiscent of a Gaussian approximation, and it is possible to show that it outperforms a confidence region based on a Gaussian approximation \cite{almeida2023comparison} for a sufficiently large number of qubits (App.~\ref{ap:analytical-comparison}).

\paragraph{Conclusions}%
    We have constructed confidence regions for quantum state tomography which can be used for any measurement set, are easy to describe and to report, are asymptotically tight and can perform better than the best regions to date in simulated experiments for high-dimensional systems.
    For applications, any improvement in the constants appearing in the respective Bernstein bounds would directly lead to an improved performance of the regions. Given the nature of the problem, only marginal gains should be expected, but even a linear improvement on the sample cost can turn an unfeasible experiment feasible.
    From a theoretical standpoint, the fact that our regions are based on the Hilbert--Schmidt distance makes them particularly versatile and amenable to analytical treatment.

\paragraph{Acknowledgements}%
    We thank O.\ Dowson for help with JULIA code optimisation, and D.\ Gross, O. Gühne, M.\ Pl\'avala, and J.\ Steinberg for discussions. 
    The OMNI cluster of the University of Siegen was used in this project.
    We acknowledge financial support from the Deutsche Forschungsgemeinschaft (DFG, German Research Foundation, project numbers 447948357 and 440958198), the Sino-German Center for Research Promotion (Project M-0294), the ERC (Consolidator Grant 683107/TempoQ), and the German Ministry of Education and Research (Project QuKuK, BMBF Grant No.\ 16KIS1618K).
    C.G.\ acknowledges support from the House of Young Talents of the University of Siegen.

\bibliography{bibliography}
\let\addcontentsline\oldaddcontentsline

\begin{appendix}

\onecolumngrid
\newpage
\newgeometry{margin=7.5em}

\thispagestyle{empty}
\begin{center}
    \textbf{\large Appendix for ``User-friendly confidence regions for quantum state tomography''}
\end{center}

\tableofcontents

\section{Confidence regions from vector Bernstein inequalities}
\label{ap:proofs}
This Section is organised as follows. In Section~\ref{ap:vector-bernstein} we present a vector Bernstein inequality which will serve as the basis for constructing the confidence regions. Subsequently, in Section~\ref{ap:bernstein-multinomial}, we consider the case of a vector-valued random variable coming from multi-Bernoulli trials and show how to bound the necessary parameters so that the vector Bernstein inequality can be used easily. This is then applied to the case of a sum of multinomial trials, where we also discuss the optimal strategy for distributing a fixed number of samples among the trials. Lastly, in Section~\ref{ap:crs}, we show how to model quantum state tomography experiments in this framework, arriving at the two confidence regions discussed in the main text.

\subsection{Vector Bernstein inequality}
\label{ap:vector-bernstein}

We start from a matrix Bernstein inequality with intrinsic dimension, here presented as in Theorem~7.7.1 in Ref.~\cite{tropp2015introduction}.

\begin{theorem}\label{thm:matrix-bernstein}
    Consider independent zero-mean Hermitian matrix-valued random variables $X_1,\dotsc,X_N$ satisfying
    \begin{equation}
        V \succcurlyeq \sum_i \expv{X_i^2}\quad\text{and}\quad
        \quad L\eye \succcurlyeq X_i
    \end{equation}
    for some Hermitian matrix $V$ and some positive number $L$.
    Define $v= \norm{V}_\infty$, $d= \tr(V)/v$, and $w(x)=3/(3+x)$.
    Then
    \begin{equation}\label{eq:matrix-bernstein}
        \prob{ \sum_i X_i \not\prec t\eye }
        \le 4d \exp \left[ -\frac{t^2}{2v}\, w \left( \frac {t L}v \right) \right]
    \end{equation}
    holds for all $t > 0$.
\end{theorem}

We specialise Theorem~\ref{thm:matrix-bernstein} to vectors instead of matrices.
The transition to such a type of statement can be found at several places in the literature, for example in Refs.~\cite{tropp2015introduction, minsker2017some, wainwright2019high}. For completeness, we also provide a proof.

\begin{corollary}\label{cor:vector-bernstein}
    Assuming independent zero-mean complex vector-valued random variables $X_1, \dotsc, X_N$ and positive numbers $\sigma$, $\lambda$ such that
    \begin{equation}
        \sigma^2 \geq \frac{1}{N} \sum_i \expv{ \norm*{X_i}_2^2 } \quad\text{and}\quad
        \lambda \geq \norm*{X_i}_2,
    \end{equation}
    we have
    \begin{equation}\label{eq:bernstein-nice}
        \prob{ \norm*{ \frac{1}{N} \sum_i X_i }_2 \geq \epsilon\sigma}
        \leq 8 \exp \left[ -\frac{N\epsilon^2}2\,w \left( \epsilon \frac \lambda\sigma \right) \right],
    \end{equation}
    for all $\epsilon > 0$.
\end{corollary}
\begin{proof}
    We define the Hermitian dilation
    \begin{equation}
        D\colon \vec x\mapsto
        \begin{pmatrix}
            0      & \vec x^\dag \\
            \vec x & 0
        \end{pmatrix}.
    \end{equation}
    and apply Theorem~\ref{thm:matrix-bernstein} to $Y_i=D(X_i)$. Consequently, we need to provide upper bounds $V\succcurlyeq \sum_i \expv{Y_i^2}$ and $L\eye\succcurlyeq Y_i$.
    First, we can use $L=\lambda$ since the largest eigenvalue of $D(\vec x)$ is given by $\norm{\vec x}_2$.
    Second, with $V_1= N\sigma^2$ and $V_2= \sum_i\expv{ X_i X_i^\dag}$, we have
    \begin{equation}
        V=
        \begin{pmatrix}
            V_1 & 0 \\
            0   & V_2
        \end{pmatrix}\succcurlyeq
        \sum_i\mathbb E
        \begin{pmatrix}
            X_i^\dag X_i & 0 \\
            0   & X_i X_i^\dag
        \end{pmatrix}=\sum_i \expv{Y_i^2}
    \end{equation}
    and $\norm{V_2}_\infty\le \sum_i \expv{\norm{X_i X_i^\dag}_\infty}\le V_1$, where the first inequality is due to the triangular inequality and the second inequality follows from $\norm{\vec x \vec x^\dag}_\infty = \norm{\vec x}_2^2$.
    This leads to $v= \norm V_\infty= \max\set{V_1,\norm{V_2}_\infty}= V_1$ and $d=\tr(V)/v=1+\tr(V_2)/V_1\le 2$.
    Defining $\epsilon= t/(N\sigma)$ and using that $D(\vec x)\not\prec t\eye$ is equivalent to $\norm{\vec x}_2\ge t$, we obtain Eq.~\eqref{eq:bernstein-nice}.
\end{proof}

\begin{remark}
In Ref.~\cite{tropp2015introduction}, Theorem~\ref{thm:matrix-bernstein} is proved for $t \geq \sqrt{v} + L/3$. However, one can safely replace this condition by $t > 0$, as presented above. This follows from $\tr(V) \geq \norm{V}_\infty$ and hence $d \geq 1$ and the proof in Ref.~\cite{tropp2015introduction}, or by virtue of the following observation.
\end{remark}

\begin{lemma}\label{lem:epsnolb}
    For $v> 0$, $L>0$, and $0< t\le \sqrt v+L/3$ we have
    \begin{equation}\label{eq:lemine}
        -\frac{t^2}{2v}\,\frac 3{3+tL/v}\ge -\frac23> -\log(4).
    \end{equation}
\end{lemma}
\begin{proof}
    We abbreviate the left-hand side of Eq.~\eqref{eq:lemine} by $z$.
    One observes that $z$ is a monotonously
    decreasing function in $t$ and hence, for a lower bound, it is sufficient to consider $t=\sqrt v+L/3$. Then one finds $-2/z-3=(\sqrt v-L/3)^2/(\sqrt v+L/3)^2\ge 0$ and thus the assertion follows.
\end{proof}

\subsection{Bernstein inequality for multinomial samples}
\label{ap:bernstein-multinomial}

Our goal is to manipulate Corollary~\ref{cor:vector-bernstein} into a form that resembles a quantum state tomography experiment (see Section \ref{sec:state-tomography-and-crs} in the main text). Essentially, then, the random variables $X_i$ in Eq.~\eqref{eq:bernstein-nice} should be replaced by $G_s ( \vec{p}_s - \vec{f}_s )$, with the observed frequencies $\vec{f}_s$, the probabilities $\vec{p}_s$, and some suitable linear maps $G_s$. These maps will later give us freedom to represent, for example, an inverse of the measurement map.

The construction is done in three steps: First, we consider trials from multi-Bernoulli distributions and investigate how to bound $\sigma$ and $\lambda$, which are the necessary quantities to apply Corollary~\ref{cor:vector-bernstein}. Then we reformulate these results in terms of multinomial counts. Lastly, we consider what is the optimal strategy to distribute the total number of samples across several multinomial trials.

\subsubsection{Application to Bernoulli trials}
\label{ap:multinoulli-trial}

We consider the case of $N$ vector-valued random variables $X_i= G_i(\vec p_i-\vec f_i)$ with $i=1,\dotsc, N$.
Here, each $\vec{f}_i=(f_{a|i})_a$ is multi-Bernoulli distributed according to $\vec{p}_i=(p_{a|i})_a$ and each $G_i$ is a linear map.
By multi-Bernoulli we denote a trial where exactly one of the outcomes labelled by $a$ can occur, so that $f_{a|i}\in \set{0,1}$ with $\sum_a f_{a|i}=1$, and the probability that outcome $a$ occurs is $p_{a|i}$. Thus, $\vec p_i= \expv{\vec f_i}$.
This generalises the situation of a biased coin or a loaded die.
From $\expv{ f_{a|i} f_{b|i} } = p_{a|i} \delta_{a,b}$ we obtain
\begin{equation}\label{eq:sigma2val}
    S_i = \expv{ \norm{X_i}_2^2 }
    = \vec{p}_i \cdot \diag(G_i^\dag G_i) - \norm*{ G_i \vec{p}_i }_2^2.
\end{equation}
This will later aid us in bounding $\sigma$. Likewise, to bound $\lambda$, it is helpful to compute
\begin{equation}
    \Lambda_i
    = \max_{\vec{f}_i}(\norm{X_i}_2^2)
    = \norm*{G_i \vec{p}_i}_2^2 + \max_a\big( \big[ \diag(G_i^\dag G_i) - 2\Re(G_i^\dag G_i\vec p_i ) \big]_a \big),
\end{equation}
where we used $f_{a|i}f_{b|i}=f_{a|i}\delta_{a,b}$.
We subsequently use the matrix norm
\begin{equation}
    \norm*{G}_{2,\infty} = \sqrt{\max_a([\diag(G^\dag G)]_a)},
\end{equation}
standing for the largest Euclidean norm among all column vectors of $G$.

We are now interested in eliminating the dependency on $\vec p_i$ by computing suitable upper bounds on $S = \sum_i S_i$ and $\Lambda= \max_i \Lambda_i$, so that we can obtain good values for $\sigma \geq \sqrt{S/N}$ and $\lambda \geq \sqrt{\Lambda}$ in Corollary~\ref{cor:vector-bernstein}.
If we assume that the set of admissible probability vectors $( \vec{p}_i )_i$ is convex --- for example, that the probabilities are from a quantum state tomography experiment --- then the resulting optimisation for $S$ can be solved with convex optimisation techniques.
We propose that in many situations it is sufficient to use $\sigma= \sqrt{S_\mathrm{max}/N}$, where
\begin{equation}
    S_\mathrm{max}= \sum_i \norm*{G_i}^2_{2,\infty}.
\end{equation}
This upper bound follows by omitting the second term in $S_i$ and relaxing the constraints on $\vec{p}_i$ to be $p_{a|i} \geq 0$ and $\sum_a p_{a|i}= 1$.
The intuition for ignoring the second order terms is that if the entries of $\diag(G_i^\dag G_i)$ are similar, then the maximum will be roughly attained for $p_{a|i}=1/\ell_i$ with $\ell_i$ the number of entries in $\vec p_i$.
The first term in $S_i$ will hence be much larger than the second term, if $\ell_i$ is large.

For $\Lambda_i$ we can use a rather loose upper bound, because it only weakly impacts the right-hand side of Eq.~\eqref{eq:bernstein-nice}.
Here, we propose to use $\lambda=\sqrt{\Lambda_\mathrm{max}}$, where
\begin{equation}
    \Lambda_\mathrm{max}= 4\max_i \norm*{G_i}_{2,\infty}^2.
\end{equation}
This follows as an upper bound, because the convexity of $\norm{G_i\vec p}_2^2$ in $\vec p_i$ gives us $\norm{G_i\vec p_i}_2^2\le \norm{G_i}_{2,\infty}^2$ and due to
\begin{equation}
    -\Re(G_i^\dag G_i\vec p_i)_a
    \le \abs{(G_i\vec e_a)^\dag(G_i\vec p_i)}
    \le \norm{G_i\vec e_a}_2\norm{G_i \vec p_i}_2
    \le \norm{G_i}_{2,\infty}^2.
\end{equation}
Here, $\vec e_a$ is the $a$th canonical vector such that $(G_i^\dag G_i\vec p_i)_a=(G_i\vec e_a)^\dag (G_i\vec p_i)$ and in the second step we used the Cauchy-Schwarz inequality.

\subsubsection{Transition to multinomial samples}

With $\sigma$ and $\lambda$ bounded for sums of multi-Bernoulli trials, we now consider multinomial distributions.
The multinomial distribution $\mathcal D(\vec p;n)$ arises from the outcomes in $n$ independent trials of the same multi-Bernoulli experiment.
That is, if $X_i$ are multi-Bernoulli distributed according to $\vec p$, then
\begin{equation}
    \left(\sum_{i=1}^n X_i\right) \dist \mathcal D(\vec p;n).
\end{equation}
Correspondingly, we assume that for our $N$ multi-Bernoulli variables $\vec f_i$, one can partition $i = 1,\dotsc,N$ into $m$ groups of size $n_s$ for $s= 1,\dotsc,m$, such that all $G_i$ and $\vec p_i$ in a group are the same.
Summing up the random variables $X_i$ in each group yields $m$ random variables $Y_s= n_s G_s(\vec p_s - \vec f_s)$ with $\vec f_s \dist \mathcal D(\vec p_s; n_s) / n_s$.
Writing
\begin{equation}
    N = \sum_s n_s,\quad
    q_s= \frac{n_s}{N},\quad
    \gamma_s=\norm*{G_s}_{2,\infty},\quad
    \gamma_\mathrm{max}= \max_s \gamma_s,\quad
    \sigma^2 = \sum_s q_s\gamma_s^2,
\end{equation}
we obtain from Corollary~\ref{cor:vector-bernstein} that
\begin{equation}\label{eq:bern-mulnom-orig}
    \prob{ \norm*{\sum_s q_s G_s(\vec p_s - \vec f_s)}_2 \geq \epsilon\sigma}
    \leq 8\exp \left[ -\frac{N\epsilon^2}2 \,
     w\left( \epsilon \frac{2\gamma_\mathrm{max}}\sigma \right) \right].
\end{equation}

\subsubsection{Sampling strategy}
\label{ap:sampling-strategies}

We consider now the design of an experiment where one can redistribute the number of samples per setting under the constraints that $N$, $\tilde G_s=G_s q_s$, and $\tilde \epsilon= \epsilon\sigma$ are all fixed.
This fixes also the confidence region function
\begin{equation}
    C\colon (\vec f_s)_s\mapsto \Set{(\vec p_s)_s|\norm*{\sum_s \tilde G_s(\vec p_s - \vec f_s)}_2 \le \tilde \epsilon}
\end{equation}
and we aim to choose $q_s$ such that $\prob{  C \not\ni (\vec p_s)_s  }$ is minimised.
Using the upper bound in Eq.~\eqref{eq:bern-mulnom-orig}, we minimise $3\sigma^2+2\tilde\epsilon\gamma_\mathrm{max}$ while keeping $\tilde\gamma_s= \norm{\tilde G_s}_{2,\infty}=\gamma_s q_s$ constant.
The Cauchy-Schwarz inequality yields
\begin{equation}
    \sigma^2=
    \sum_s \Big(\frac{\tilde\gamma_s}{\sqrt{q_s}}\Big)^2
    \sum_s \sqrt{q_s}^2\ge \Big(\sum_s \tilde\gamma_s\Big)^2
\end{equation}
and equality holds for $q'_s=\tilde\gamma_s/\sum_t\tilde\gamma_t$.
At the same time, we have $\gamma_\mathrm{max}\ge \sigma\ge \sum_s\tilde\gamma_s$ with equality also at $q'_s$.
Hence it is optimal to distribute the number of samples according to $n_s=q_s'N$.

\subsubsection{Summary}
\label{ap:summary-strategies}

\begin{mybox}{Multinomial estimation}
    For $s = 1,\dotsc,m$ assume matrices $G_s \in \mathbb C^{d \times \ell_s}$ and numbers of trials $n_s\in \mathbb N$.
    Let $\gamma_s=\norm{G_s}_{2,\infty}$, $N=\sum_sn_s$, $q_s=n_s/N$,
    \begin{equation}\label{eq:sigmaeta}
        \sigma = \sqrt{\sum_s \gamma_s^2/q_s}, \quad\text{and}\quad \eta=\frac2\sigma\max_s(\gamma_s/q_s).
    \end{equation}
    For probability vectors $\vec p_s \in \mathbb R^{\ell_s}$ and multinomial distributions $\mathcal D(\vec p_s;n_s)$ assume independent random variables $\vec f_s \dist \mathcal D(\vec p_s;n_s)/n_s$. Then
    \begin{equation}\label{eq:mulest}
        \prob{ \norm*{\sum_s G_s(\vec p_s - \vec f_s)}_2 \ge \epsilon \sigma }
        \le 8\exp \left[ -\frac{N\epsilon^2}2 \frac{3}{3+\eta\epsilon}\right].
    \end{equation}
    It is optimal to have $q_s=\gamma_s/\sum_t \gamma_t$, where $\sigma=\sum_s\gamma_s$ and $\eta=2$.
\end{mybox}

From the above we obtain the confidence region
\begin{equation}
    C(\epsilon;\vec f_1,\dotsc,\vec f_m)=
    \Set{(\vec p_1,\dotsc,\vec p_m)|\norm*{\sum_s G_s(\vec p_s - \vec f_s)}_2 \le \epsilon\sigma},
\end{equation}
where for an intended confidence level $1 - \delta$ one has to choose
\begin{equation}\label{eq:epsilon}
    \epsilon = \frac6\eta\sqrt{u}(\sqrt u+\sqrt{u+1}) \quad \text{with}\quad u = \frac{\eta^2}{18 N} \log(8/\delta).
\end{equation}
Notice that $\epsilon$ is approximately proportional to $1/\sqrt N$ and $\sqrt{2 \log(8/\delta)}<\epsilon\sqrt N< \sqrt{3\log(8/\delta)}$ for $N>\frac{4}3\eta^2\log(8/\delta)$.

\subsection{Two confidence regions for quantum state tomography}
\label{ap:crs}

We recall that a quantum state tomography experiment boils down to $N$ systems prepared in the same unknown state $\rho$, from which we obtain the frequencies $f_{a|s}$ representing the occurrences of the outcome $a$ for the measurement setting labelled by $s$.
For $n_s$ measurements of setting $s$, these frequencies are distributed according to a multinomial distribution, $\vec f_s \dist \mathcal D(\vec p_s; n_s)/n_s$, with $\vec p_s = ( \tr(\rho E_{a|s}) )_a$ and where $E_{a|s}$ is the measurement operator for outcome $a$. We combine all probabilities to $\vec p=\bigoplus_s \vec p_s$ and similarly, for the frequencies $\vec f_s = (f_{a|s})_a$, we write $\vec f=\bigoplus_s \vec f_s$.

\subsubsection{Construction of the regions}\label{ap:2regions}
It is convenient to define the total measurement map $M\colon \rho \mapsto \vec p$ as a linear map from the space of self-adjoint operators to the real vector space spanned by all $\vec p$.
For tomography, the map $M$ is assumed to be injective and hence one can choose a left-inverse $M^+$ of $M$, that is, $M^+$ is such that $M^+M A=A$ for any self-adjoint operator $A$.
A common choice for $M^+$ is the pseudoinverse of $M$.
For data $\vec f$, we define the point estimate $\hat\rho(\vec f)= M^+\vec f$, which yields the fiducial state on average, that is, $\expv{\hat\rho}=\rho$.
Finally, we write $M^+_s$ for the component of $M^+$ that acts on $\vec p_s$, such that $M^+ \vec p = \sum_s M^+_s \vec p_s$.

For the first confidence region, we observe that $\sum_s M^+_s(\vec p_s-\vec f_s)= \rho-\hat\rho(\vec f)$. This suggests to let $G_s = M^+_s$ in our multinomial estimation procedure in Section~\ref{ap:summary-strategies}, yielding the region
\begin{equation}
    C_A(\epsilon; \vec f)
    = \set{ \rho \mid \norm*{\rho-\hat\rho(\vec f)}_\mathrm{HS} \leq \epsilon \sigma_A},
\end{equation}
where $\norm{A}_\mathrm{HS}= \sqrt{\tr(A^\dag A)}$ denotes the Hilbert--Schmidt norm on the self-adjoint operators.
The second region is obtained by noticing that $\sum_s MM^+_s(\vec p_s - \vec f_s) = M[\rho - \hat\rho(\vec f)]$ and thus we may alternatively set $G_s = MM^+_s$. This leads to
\begin{equation}
    C_B(\epsilon;\vec f)
    = \set{\rho \mid \norm*{\rho-\hat\rho(\vec f)}_M\le \epsilon \sigma_B },
\end{equation}
where $\norm{A}_M = \norm{M A}_2$.
For either of the regions, $\epsilon$ is determined by the intended confidence level $1 - \delta$ via Eq.~\eqref{eq:epsilon} where the parameters $\sigma_{A/B}$ and $\eta_{A/B}$ are computed from Eq.~\eqref{eq:sigmaeta}.

The characteristic difference between $C_A$ and $C_B$ is in their shape.
While $C_A$ describes a sphere with radius $\epsilon\sigma_A$ in state space, $C_B$ describes an ellipsoid with semiaxes $\epsilon A_j$, both with respect to the Hilbert--Schmidt norm.
The operators $A_j$ are given by
\begin{equation}\label{eq:genhalfax}
    A_j = \frac{\sigma_B}{\sqrt{\xi_j}}\, \Xi_j,
\end{equation}
where the numbers $\xi_j$ and the self-adjoint operators $\Xi_j$ form a spectral decomposition of $M^\dag M$, that is, $M^\dag M\Xi_j=\xi_j\Xi_j$ and $\tr(\Xi_j\Xi_{j'})=\delta_{j,j'}$.

\subsubsection{Optimising the left-inverse}\label{ap:optim-leftinverse}
The confidence regions $C_A$ and $C_B$ depend on the left-inverse $M^+$ of the measurement map $M$. If $M$ is singular, then the left-inverse is not unique and a common choice is the pseudoinverse. But this is not always optimal. For both confidence regions, the choice of $M^+$ determines the centre of the region via $\hat \rho(\vec f)=M^+\vec f$ and, crucially, it also influences the size of the region via $\sigma_{A/B}$.
To avoid confusion, in this section we write $M^+$ for the general left-inverse and $M^-$ for the pseudoinverse.

We first consider the optimal choice for region $C_B$ and assume that the sampling strategy is later chosen such that $\sigma_B=\sum_s \norm{MM^+_s}_{2,\infty}$ with $M^+=(M^+_1,M^+_2,\dotsc)$.
A left-inverse that achieves the minimum $\hat\sigma_B$ can be readily found by solving the second order cone optimisation
\begin{equation}\label{ap:primal-opt-linv}
    \hat \sigma_B= \min_X \sum_s\norm{Q_s+MXK_s}_{2,\infty},
\end{equation}
where $Q_s=MM^-_s$, the row vectors of $K=(K_1,K_2,\dotsc)$ span the kernel of $M^\dag$, and $X$ is any appropriately shaped matrix. This is the case, because any left-inverse of $M$ can be written as $M^-+XK$. We compute a general lower bound via the Lagrange dual problem \cite{boyd}, yielding
\begin{equation}\label{ap:dual-opt-linv}
    \hat\sigma_B\ge \max_{(Y_s)_s}\set{ \sum_s\tr(Q_sY_s) | \sum_sK_s Y_s M=0 \text{ and } \norm{Y_s^\dag}_{2,1}\le 1},
\end{equation}
where $\norm{A}_{2,1}$ is the sum of the column norms of $A$ and $Y_s$ are any appropriately shaped matrices.
For $Y_s=Q_s^\dag/\max_t(\norm{Q_t}_{2,1})$ and $X=0$, we obtain the bounds
\begin{equation}
    \frac{\sum_s \norm{Q_s}_{2,2}^2}{\max_s(\norm{Q_s}_{2,1})}
    \le\hat\sigma_B\le
        \sum_s\norm{Q_s}_{2,\infty}.
\end{equation}
Hence, the pseudoinverse is optimal, if all column norms of $MM^-_s$ are equal to $\nu_s$ and all $\ell_s\nu_s$ are equal, where $\ell_s$ is the number of outcomes for setting $s$.

The analysis for the confidence region $C_A$ and $\sigma_A$ is similar, because the duality relation between Eq.~\eqref{ap:primal-opt-linv} and Eq.~\eqref{ap:dual-opt-linv} holds for arbitrary matrices $Q_s$, $M$, and $K_s$. Hence we replace in both equations $M$ by the identity map and $Q_s$ by $M^-_s$. Then the pseudoinverse is optimal if all column norms of $M^-_s$ are equal to $\mu_s$ and all $\ell_s\mu_s$ are equal.

In practice, the pseudoinverse is often the appropriate choice. First, the above sufficient conditions for optimality are satisfied by the large class of local symmetric measurements, see Section~\ref{ap:sic-v3}. Second, we expect that for most other practically relevant tomographic measurements, either this is also the case, or an optimised pseudoinverse gives only a minor improvement in $\sigma_{A/B}$.

\section{Construction of confidence regions for quantum state tomography}
\label{ap:explicit-construction}

In this section we give a general construction scheme, see Section~\ref{ap:construct-general}, and explicit examples for the confidence regions $C_A$ and $C_B$.
Specifically, in Section~\ref{ap:sic-v3}, we analyse the large class of local symmetric qudit measurements.
These results are then specialised to the Pauli-bases measurement in Section~\ref{ap:pauli-tomography} and measurement of all Pauli observables in Section~\ref{ap:pauli-obs2}, both for an arbitrary number of qubits.

\subsection{General scheme}
\label{ap:construct-general}

Here we provide hands-on instructions to construct the confidence regions $C_A$ and $C_B$ introduced in Section~\ref{ap:crs}.
The analysis so far allows us to treat measurements settings individually by means of the results summarised in Section~\ref{ap:summary-strategies}, offering an optimised performance of the confidence region.
But, from a practical perspective, the generality of this analysis complicates the actual calculations much more than the typical gain in performance could justify.

This motivates us to employ a common simplification, where one assumes that the tomographic scheme consists of a single generalised measurement described by a positive operator-valued measure, that is, a collection of positive semidefinite operators $E_a\succcurlyeq 0$ with $\sum_a E_a=\eye$.
This is no loss of generality, in the following sense.
If one has several measurements settings $s=1, \dotsc, m$, one assumes that they are performed at random with probability $q_s$ each.
Then, if the operator $E'_{a'|s'}$ describes the measurement outcome $a'$ of measurement setting $s'$, the unified measurement has outcomes $a=(a',s')$ described by the operators $E_a=q_{s'} E'_{a'|s'}$. Correspondingly, one has to replace the observed frequency $f'_{a'|s'}$ of outcome $a'$ for setting $s'$ by $f_a=q_{s'}f_{a'|s'}$.

For the upcoming computations, it is convenient to use a vectorisation $\vec v(A)\in \mathbb R^{d^2}$ of the self-adjoint operators $A$ on a $d$-dimensional Hilbert space with the property that $\vec v(A)\cdot \vec v(B)= \tr(A B)$. A simple example of such a vectorisation is given by concatenating the rows of $\Re(A)+\Im(A)$ of a matrix-representation of $A$; for the inverse $\vec v^{-1}$, one partitions $\vec v(A)$ such that it yields a $d\times d$ matrix $V$ and then computes $A=\frac12(V+V^\intercal)+\frac i2(V-V^\intercal)$.

\medskip
\begin{mybox}{Construction steps}
    We construct the confidence regions $C_A$ and $C_B$ for the confidence level $1-\delta$, assuming that the generalised measurement $(E_1,E_2,\dotsc)$ has been preformed a total of $N$ times, yielding the observed frequencies $\vec f=(f_1,f_2,\dotsc)$.
    \smallskip
    \begin{enumerate}[(1)]
        \item
        Determine $\epsilon$ as
        \begin{equation}
            \epsilon=3\sqrt{u}(\sqrt u+\sqrt{u+1})
            \quad\text{with}\quad u=\frac{2}{9N}\log(8/\delta).
        \end{equation}
        \item
        Construct $M$ as the matrix where row $a$ is ${\vec v(E_a)}$.
        \item
        Compute the pseudoinverse $M^+$ of $M$ and let $\hat\rho=\vec v^{-1}(M^+ \vec f)$.
        \item
        For $C_A$: Compute $\sigma_A$ as the maximal Euclidean norm of the columns of $M^+$.
        \item
        For $C_B$: Compute $\sigma_B$ as the maximal Euclidean norm of the columns of $MM^+$.
        \item
        For $C_B$: Compute the spectral decomposition of ${M}^\intercal M$, yielding the eigenvalues $\xi_j$ and corresponding orthonormal eigenvector $\vec x_j$ and let
        \begin{equation}
            A_j=\frac{\sigma_B}{\sqrt{\xi_j}}
            \,\vec v^{-1}(\vec x_j).
        \end{equation}
    \end{enumerate}
\end{mybox}
\medskip

Then, the confidence region
\begin{itemize}
    \item $C_A$ is given by all $\rho$, such that
    \begin{equation}
      \norm{\rho-\hat\rho}_\mathrm{HS}\le \epsilon\sigma_A,
    \end{equation}
    where $\norm{A}_\mathrm{HS}=\sqrt{\tr(A^\dag A)}$ is the Hilbert--Schmidt norm.
    $C_A$ is a sphere in the Hilbert--Schmidt norm, centred at $\hat \rho$ and with radius $\epsilon\sigma_A$.
    \item $C_B$ is given by all $\rho$, such that
    \begin{equation}
        \norm{M \vec v(\rho-\hat \rho)}_2\le \epsilon \sigma_B,
    \end{equation}
    where $\norm{\vec x}_2$ is the Euclidean norm.
    $C_B$ is an ellipsoid in the Hilbert--Schmidt norm, centred at $\hat \rho$ and with semiaxes $\epsilon A_j$.
\end{itemize}

\begin{remark}
    We argued that one can combine all measurement setting into a single generalised measurement.
    However, this assumes that the numbers of samples per setting $n_s= q_sN$ are also random, that is, they are a sample of the multinomial distribution $\mathcal D(\vec q,N)$.
    One may wonder whether the above construction methods are also valid, if one chooses $n_s$ beforehand, so that they are fixed rather than random.
    This is intuitively true, because choosing $n_s$ at random cannot reduce the statistical fluctuations in the data and hence letting $q_s=n_s/N$ should not require a larger confidence region.
    One can also see this by using formal manipulations of the result in Section~\ref{ap:summary-strategies}.
    We assume that the single measurement with operators $E_{a,s}$ is composed from different settings via $E_{a,s}=q_sE'_{a|s}$ with $\sum_a E'_{a|s}= \eye$.
    Correspondingly, we define $p'_{a|s} = \tr(E'_{a|s}\rho)= \tr(E_{a,s}\rho)/q_s=p_{a,s}/q_s$, assume that the data $f'_{a|s}$ is obtained by sampling from $\mathcal D(\vec p'_s,n_s)/n_s$, and let $f_{a,s}=q_s f'_{a|s}$.
    Writing $G=(G_1,\dotsc,G_m)$, we define $G_s'=q_sG_s$, where $G=M^+$ for $C_A$ and $G=MM^+$ for $C_B$.
    Then $G_s(\vec p_s-\vec f_s)=G'_s(\vec p'_s-\vec f'_s)$ and Eq.~\eqref{eq:mulest} applies with
    \begin{equation}
        \sigma^{\prime2}=\sum_s \norm{G'_s}_{2,\infty}^2/q_s=\sum_s q_s\norm{G_s}_{2,\infty}^2\le \max_s \norm{G_s}_{2,\infty}^2=\norm{G}_{2,\infty}^2=\sigma^2.
    \end{equation}
    Hence, we can use $\sigma$ instead of $\sigma'$ and then $\eta'=2\max_s(\norm{G'_s}/q_s)/\sigma=2$, yielding the same confidence region that one obtains from the formulation with a single measurement.
\end{remark}

\subsection{Local symmetric measurements}\label{ap:sic-v3}

A convenient measurement scheme for $q$ qudits consists of measuring on each qudit the same local generalised measurement with $n\ge d^2$ outcomes. If this local measurement is highly symmetric, one can obtain explicit expressions for the radius of the confidence region $C_A$ and the semiaxes of the confidence region $C_B$.

Inspired by the analysis in \cite{graydon2016quantumconicaldesigns} and \cite[App.
 B]{guta2020fast}, we consider a measurement with $n$ measurement operators $E_a= \vartheta \Pi_a$, where $a=1,\dotsc, n$ and $\Pi_a$ are projections of rank $r$.
These operators shall satisfy the symmetry condition
\begin{equation}\label{eq:symcond}
    \sum_a E_a\otimes E_a = \alpha S+\beta\eye,
\end{equation}
for the swap operator $S\ket{\phi}\ket\psi=\ket\psi\ket\phi$.
This is equivalent to require that the square $F$ of the measurement map $M\colon A\mapsto (\tr(E_a A))_a$ obeys
\begin{equation}
    FA= M^\dag MA =\sum_a E_a\tr(E_aA)=\sum_a \tr_1\left((E_a\otimes E_a)(A\otimes \openone)\right) = \alpha A+\beta \eye \tr(A),
\end{equation}
where we used $\tr_1(A\otimes B)=B\tr(A)$ and $\tr_1(S(A\otimes B))=AB$.
Although we did not fix $\vartheta$, $\alpha$, and $\beta$, they are necessarily given by
\begin{equation}
    \vartheta=\frac d{nr},\quad
    \alpha = \vartheta \frac{d-r}{d^2-1}\quad\text{and}\quad
    \beta = \vartheta-\alpha d > 0.
\end{equation}
This can be seen by evaluating the trace of $\sum_a E_a=\eye$, yielding $\vartheta r n=d$, then taking the trace of Eq.~\eqref{eq:symcond} to get $\vartheta^2 r^2 n= \alpha d + \beta d^2$, and finally, by realizing that $\tr(\Pi_a\otimes \Pi_a S)=\tr(\Pi_a^2)= r$, so that after multiplying both sides of Eq.~\eqref{eq:symcond} by $S$, the trace yields $\vartheta^2 r n=\alpha d^2 + \beta d$.

The inverse of $F$ is readily verified to be given by
\begin{equation}
    F^{-1}A= \frac 1\alpha A-\frac{n\beta}{d\alpha}\eye\tr(A).
\end{equation}
This enables us to compute the pseudoinverse $M^+=F^{-1}M^\dag$.
We write $Z_a$ for column $a$ of $M^+$, that is,
\begin{equation}
    Z_a=F^{-1}E_a=\frac1\alpha E_a-\frac\beta\alpha \eye.
\end{equation}
First, we obtain for the estimator $\hat\rho(\vec f) = M^+\vec f$ the explicit expression
\begin{equation}
    \hat \rho(\vec f)=\frac{d^2-1}{d-r}\sum_a f_a\Pi_a-\frac{dr-1}{d-r}\eye.
\end{equation}
Second, the diagonal entries of $M^{+\dag}M^+$ are given by $\tr(Z_a^2)$. These are clearly independent of $a$, yielding the radius $\epsilon\sigma_A$ of the confidence region $C_A$ with
\begin{equation}
    \sigma_A= \sqrt{\frac rd \frac {(d^2-1)^2}{d-r}+ \frac 1d}.
\end{equation}
For $C_B$, we use that $MM^+$ is an orthogonal projection, so that
$\diag[(MM^+)^\dag(MM^+)]_a= \tr(E_a Z_a)=d^2/n=\sigma_B^2$ and that the eigenvalues of $F$ are $d/n$ in direction $\eye$ and $\alpha$ for all zero-trace operators.
Hence the semiaxes of $C_B$ are $\epsilon A_\mu$ with
\begin{equation}
    A_\mu=\left(r \frac{d^2-1}{d-r}\right)^{(1-\delta_{\mu,0})/2}\sqrt d\,B_\mu
    \quad\text{for } \mu=0,\dotsc,d^2-1,
\end{equation}
where the operators $B_\mu$ form any orthonormal basis of the self-adjoint operators with $B_0=\eye/\sqrt d$.

Next, we consider a system of $q$ qudits, where each qudit is subjected to the same symmetric measurement.
We write $E_{\vec a}=\bigotimes_k E_{a_k}$ for the global measurement operators, where $\vec a=(a_1,a_2,\dotsc,a_q)$. The total measurement map $M\colon \rho\mapsto (\tr(E_{\vec a}\rho))_{\vec a}$ is then given by $M=M_L^{\otimes q}$, where $M_L$ is the measurement map for the single qudit.
Similarly, $F=F_L^{\otimes q}$, $M^+=(M_L^+)^{\otimes q}$, and $Z_{\vec a}=\bigotimes_k (Z_L)_{a_k}$. It follows that the confidence region $C_A$ has radius $\epsilon\sigma_A$ with
\begin{equation}
    \sigma_A= (\sigma_{A,L})^q= \left(\frac rd \frac {(d^2-1)^2}{d-r}+ \frac 1d\right)^{q/2}.
\end{equation}
The semiaxes of $C_B$ are $\epsilon A_{\vec \mu}$ with
\begin{equation}
    A_{\vec \mu}
    = \bigotimes_{k=1}^q A_{\mu_k}
    = \left(r \frac{d^2-1}{d-r}\right)^{(q-\chi_{\vec \mu})/2}d^{q/2}
    \bigotimes_{k=1}^q B_{\mu_k},
    \quad \vec\mu\in \set{0,1,\dotsc d^2-1}^q
\end{equation}
where $\chi_{\vec\mu}$ counts the number of zeros in $\vec\mu$, that is, $\chi_{\vec\mu}=\sum_k \delta_{\mu_k,0}$.

We observe that the sizes of $C_A$ and $C_B$ do not depend on the specific choice of the local measurement and their number of outcomes, as long as they are symmetric as specified above.
This can be surprising, since the symmetric measurements cover a wide range of common measurement schemes. For unit rank and a single system, $r=1$ and $q=1$, these are the structured measurements discussed in Ref.~\cite{guta2020fast}. Examples are 
symmetric informationally-complete measurements \cite{renes2004sic}, the Pauli-bases measurement (see also Section~\ref{ap:pauli-tomography}), and the uniform measurement with operators proportional to $\dyad{\psi}\dd\psi$.
In the case of qubits, $d=2$, we can express the conditions for a structured measurement using Bloch vectors $\vec x_a=(\tr(\Pi_a\sigma_j)_j$ as $\vec x_a\cdot\vec x_a=1$, $\sum_a \vec x_a=\vec 0$, and $3\sum_a \vec x_a \vec x_a^\intercal= n \eye$, that is, $(\vec x_a)_a$ is a tight unit norm frame with mean $\vec 0$.
Another important case of symmetric measurements is the case of Pauli observables on $Q$ qubits, see Section~\ref{ap:pauli-obs2}, for which $d=2^Q$, $r=d/2$, and $q=1$.

\begin{remark}
    We may reconsider $S_i$ in Eq.~\eqref{eq:sigma2val} for a possible improvement and compute exact upper bounds on $S_i$ in order to test whether $S_\mathrm{max}$ is a good approximation.
    For this one first realises that $\norm{G^{A/B}\vec p}_2^2$ is minimised for the completely mixed state with $p_a= n^{-q}$.
    Then one finds $\max_\rho(S_i^A) = S_\mathrm{max}^A-d^{-q}$ and $\max_\rho(S_i^B) = S_\mathrm{max}^B-n^{-q}$, yielding only a negligible improvement and hence our approximation to use $S^{A/B}_\mathrm{max}$ is sufficiently good. In addition, one may also consider whether the pseudoinverse is a good choice in the light of the discussion in Section~\ref{ap:optim-leftinverse}. This is indeed the case, because the columns of $G^{A/B}$ have all the same norm, respectively.
\end{remark}

\subsection{Pauli-bases measurement}
\label{ap:pauli-tomography}

In the Pauli-bases measurement scheme, each qubit is measured using one of the Pauli operators $\sigma_s$, where $s \in \set{1, 2, 3}$ represents $X$, $Y$ or $Z$.
Each of these measurements returns one of two possible outcomes, labelled by $a \in \set{-1, +1}$ and associated to the projectors $\Pi_{a|s} = (\eye + a \sigma_s) / 2$.
We use $\vec s = (s_1, \ldots, s_q) \in \set{1, 2, 3}^q$ and $\vec a = (a_1, \ldots, a_q) \in \set{-1, +1}^q$ to represent the collective local settings and outcomes of the $q$ qubits, respectively.
We assume that each setting is measured with the same weight $q_{\vec s}=3^{-q}$, so that effectively the generalised measurement with the $6^q$ outcomes $E_{\vec a,\vec s}=3^{-q}\Pi_{\vec a|\vec s}$ is measured, yielding the probabilities
\begin{equation}\label{eq:pauliM}
    p_{\vec a,\vec s}
    = (M \rho)_{\vec a,\vec s}
    = 3^{-q}\tr(\rho \Pi_{\vec a|\vec s})
    \quad\text{with}\quad \Pi_{\vec a|\vec s}=\bigotimes_{k=1}^q \Pi_{a_k|s_k}.
\end{equation}

The measurement described by $E_{\vec a,\vec s}$ satisfies the conditions of a local symmetric measurement discussed in Section~\ref{ap:sic-v3} with $d=2$, $r=1$, and $n=6$.
It follows that the confidence region $C_A$ has the radius $\epsilon\sigma_A$ with
\begin{equation}
    \sigma_A=5^{q/2}.
\end{equation}
Similarly, the semiaxes of $C_B$ are $\epsilon A_{\vec \mu}$ with
\begin{equation}\label{eq:semiaxis}
    A_{\vec \mu}= 3^{ (q-\chi_{\vec \mu})/2 } \bigotimes_{k=1}^q \sigma_{\mu_k} \quad\text{for } \vec\mu\in\set{0,1,2,3}^q,
\end{equation}
where $\chi_{\vec \mu}$ counts the number of zeros in $\vec \mu$ and $\sigma_0=\eye$.

Because we mostly use the Pauli-bases measurement in our numerical analysis, we collect here a few formulas which are useful for that purpose. The map $M$ defined in Eq.~\eqref{eq:pauliM}, evaluated for $\sigma_{\vec \mu}=\bigotimes_k \sigma_{\mu_k}$ has the components
\begin{equation}\label{eq:mpaulicomp}
    (M\sigma_{\vec \mu})_{\vec a,\vec s}
    = 3^{-q}\prod_k(\delta_{\mu_k,0}+a_k\delta_{\mu_k,s_k})
\end{equation}
and squares to
\begin{equation}\label{eq:msqcoef}
    M^\dag M \sigma_{\vec \mu}=3^{\chi_{\vec \mu}-2q}\sigma_{\vec \mu}.
\end{equation}
The pseudoinverse $M^+$ is given by \cite{guta2020fast},
\begin{equation}
    M^+\vec x
    = \sum_{\vec a,\vec s} x_{\vec a,\vec s}
      \bigotimes_{k=1}^q \left(3\Pi_{a_k|s_k}-\eye\right).
\end{equation}
It has the components
\begin{equation}
    (M^{+\dag}\sigma_{\vec \mu})_{\vec a,\vec s}
    = \prod_k(\delta_{\mu_k,0}+3a_k\delta_{\mu_k,s_k})
\end{equation}
and squares to
\begin{equation}
    (M^{+\dag}M^+)_{(\vec a,\vec s),(\vec b,\vec t)}
    = 2^{-q}\prod_k(1+9a_k b_k \delta_{s_k,t_k}).
\end{equation}

\subsection{Global Pauli observables}\label{ap:pauli-obs2}
We consider the measurements of all Pauli observables $\sigma_{\vec s}=\bigotimes_k \sigma_{s_k}$, $\vec s\in \set{0,1,2,3}^Q$ on $Q$ qubits, that is, on a system with dimension $d=2^Q$.
The measurement outcomes correspond to the $\pm1$ eigenspaces of the observables, $\Pi_{\pm 1|\vec s}=\frac12(\eye\pm \sigma_{\vec s})$, and we exclude the trivial measurement with $\vec s=\vec 0$.
Hence, there are $d^2-1$ measurement settings and we assume that each setting is measured with the same weight $q_{\vec s}=(d^2-1)^{-1}$.

Hence we have a single measurement with operators $E_{\pm1,\vec s}=q_{\vec s}\Pi_{\pm1|\vec s}$.
It satisfies the conditions of a joint symmetric measurement with $q=1$, $d=2^Q$, $r=d/2$, and $n=2(d^2-1)$. Consequently, the confidence region $C_A$ has the radius $\epsilon \sigma_A$ with
\begin{equation}
    \sigma_A=\sqrt{\frac{(d^2-1)^2}d +\frac1d}
\end{equation}
and the semiaxes of $C_B$ are given by $\epsilon A_{\vec \mu}$ with
\begin{equation}
    A_{\vec \mu}=(d^2-1)^{(1-\delta_{\vec \mu,\vec 0})/2}\sigma_{\vec \mu} \quad\text{for}\quad \vec\mu\in\set{0,1,2,3}^Q.
\end{equation}
Finally, the estimator $\hat\rho$ takes a particularly simple form for the case of Pauli observables \cite{guta2020fast},
\begin{equation}
    \hat\rho(\vec f) = \frac\eye d + \frac{d^2-1}d\sum_{\vec s}(f_{+1,\vec s}-f_{-1,\vec s})\sigma_{\vec s},
\end{equation}
where the sum is over all $\vec s\in\set{0,1,2,3}^Q$, except $\vec s=\vec 0$.

\section{Performance analysis}
\label{ap:performance}

In Section~\ref{ap:analytical-comparison}, we analytically compare the confidence region $C_B$ with the Gaussian confidence region $C_2$ proposed in Ref.~\cite{almeida2023comparison}.
Subsequently, in Section~\ref{ap:improved-guta}, we provide details on the reference confidence region $C_R$. In Section~\ref{ap:gme}, we study the performance of $C_B$ in the certification of genuine multipartite entanglement.

\subsection{Comparison with a Gaussian confidence region}
\label{ap:analytical-comparison}
We compare the confidence region $C_B$ to the recently proposed confidence region $C_G$ that is based on a Gaussian approximation; see Confidence Region~$C_2$ in Ref.~\cite{almeida2023comparison}. Both regions have the same shape and only differ in their size. For a system in dimension $d$,
\begin{equation}
    \prob{\norm{\rho-\hat\rho}_M> \epsilon \sigma_G }
      \le S_{d^2-1}(N\epsilon^2)
\end{equation}
where $S_k$ is the survival function of the $\chi^2$ distribution with $k$ degrees of freedom  and $\sigma_G^2$ is the maximal eigenvalue of the covariance matrix,
\begin{equation}
    \sigma_B^2= \max_{\rho}\max_s\left(\frac1{q_s}\norm{\diag(\vec p_s) - \vec p_s \vec p_s^\intercal}_\infty\right)\le \max_{s,a}\left(\frac1{q_s}\norm{E_{a|s}}_\infty\right),
\end{equation}
where $p_{a|s}=\tr(E_{a|s}\rho)$, $q_s=n_s/N$, and $n_s$ is the number of samples for setting $s$.

For the comparison, we consider the case of local symmetric measurements, see Section~\ref{ap:sic-v3}. There we have $\max_{a} \norm{E_a}_\infty = \vartheta$ which we use for $\sigma_G$. Then the size of the Gaussian region is given by
\begin{equation}
    \epsilon_G=\epsilon\sigma_G= \sqrt{\vartheta \frac{S^{-1}_{d^{2q}-1}(\delta)}N},
\end{equation}
with $S^{-1}_k$ the inverse survival function and $1-\delta$ the confidence level.
We compare this to the size of the confidence region $C_B$, which is $\epsilon_B=\epsilon \sigma_B=\sigma_B\sqrt{2\log(8/\delta)/N}$ for large $N$. The ratio of both sizes is
\begin{equation}\label{eq:BGratio}
  \frac{\epsilon_B}{\epsilon_G}
    = \left(\frac rd\right)^{q/2}\sqrt{ \frac{2 d^{2q}\log(8/\delta)}{S^{-1}_{d^{2q}-1}(\delta)}}= O\left(\frac rd\right)^{q/2},
\end{equation}
due to $\lim_{k\to\infty} S_k^{-1}(\delta)/k=1$.
For local structured measurements and specifically the Pauli-bases measurement, we have $r=1$ and hence $C_B$ performs significantly better for large dimensions and for a large number of qudits. A different situation occurs for Pauli observables on $Q$ qubits, where $(r/d)^q=1/2$ is constant and hence $C_G$ performs better than $C_B$, albeit only by a constant factor.

For a more fair comparison, one may consider instead of a single measurement with $n$ outcomes a situation of $m$ projective measurement settings, each with $d/r$ outcomes and $q_s=1/m$. This is in particular sensible for the Pauli-bases measurement and the measurement of Pauli observables. Using the bound for projective measurements, $\sigma_G\le \max_s 1/(2q_s)=m/2$ and taking into account that $M$ is rescaled by $m$, this yields the same expression as in Eq.~\eqref{eq:BGratio} but with an extra factor of $\sqrt 2$.
We show $\epsilon_B/\epsilon_G$ for the Pauli-basis measurements in dependence of $q$ including this extra factor in Figure~\ref{fig:newold}. From that, one can infer that above $5$ qubits, $C_B$ becomes superior.

\begin{figure}
    \centering
    \includegraphics[width=.6\linewidth]{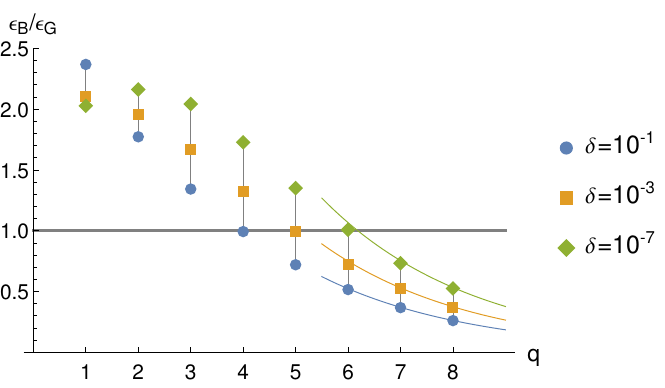}
    \caption{\label{fig:newold}%
    Comparison of the size of the confidence region $C_B$ and the confidence region $C_2$ in Ref.~\cite{almeida2023comparison} depending on the number of qubits $q$ and for different confidence levels $1-\delta$.
    The ratio $\epsilon_B/\epsilon_G$ is shown for the Pauli-bases measurement on $q$ qubits. The thin lines show the asymptotic $O(2^{-q/2})$ behaviour for large $q$.}
\end{figure}

\subsection{Amended bound for the reference confidence region}
\label{ap:improved-guta}
Theorem 6.1 in Ref.~\cite{tropp2012user} presents the following matrix Bernstein inequality, valid for all $t>0$:
\begin{equation}\label{eq:bernstein-mat}
    \prob{\sum_i X_i \not\prec \eye t } \leq d \exp\left[
    - \frac{t^2}{2v} \tilde w\left(\frac{tL}{v}\right)
    \right].
\end{equation}
Here, $X_i$ are independent zero-mean $d\times d$ Hermitian matrix-valued random variables, $v\ge \norm{\sum_i \expv{X_i^2}}_\infty$, $L\eye \succcurlyeq X_i$, and for $\tilde w(x)$ one can use either of the functions
\begin{equation}
    \tilde w_T(x)=\frac 2{x^2} \left[(x+1)\log(x+1)-x\right]
    \ge \tilde w_{S}(x)=\frac 3{3+x} \ge \tilde w_R(x)=\frac 3{4\max(1,x)},
\end{equation}
see also Figure~\ref{fig:guta-epsilon-comparison}.
Note the similarity between this Bernstein inequality and the one that we present as Theorem~\ref{thm:matrix-bernstein} in Section~\ref{ap:vector-bernstein}.

\begin{figure}
    \centering
    \includegraphics[width=.6\linewidth]{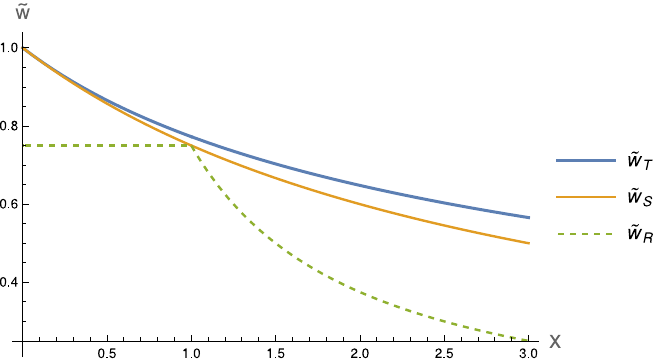}
    \caption{Functions that are commonly used for $\tilde w(x)$ in the matrix Bernstein inequality \eqref{eq:bernstein-mat}. $\tilde w_T$ gives the tightest bound and $\tilde w_S$ is a very close approximation. The most interesting regime is $x<1$, and $\tilde w_R$ is here simply a constant.}
    \label{fig:guta-epsilon-comparison}
\end{figure}

In Ref.~\cite{guta2020fast}, a confidence region for quantum state tomography is derived from Eq.~\eqref{eq:bernstein-mat} and using the function $\tilde w_R(x)$. For a fair comparison with confidence region $C_A$ and $C_B$, we provide here an amended version of this confidence region, using $\tilde w_S(x)$.
Then, from the analysis in Ref.~\cite{guta2020fast}, it readily follows that
\begin{equation}\label{eq:guta-general-nt}
    \prob{\norm*{ \rho - \hat{\rho} }_\infty \geq \epsilon\sigma_R } \leq 2 d \exp\left[
      -\frac{N \epsilon^2}{2}\, \tilde w_S(\epsilon\eta_R) \right],
\end{equation}
where $\sigma_R= \sqrt{v/N}$ and $\eta_R=\lambda/\sigma_R$ for $\lambda\ge \norm{X_i}_\infty$. For comparison with the symbols used in Ref.~\cite{guta2020fast}, we have to replace $\lambda\to RN$ and $v\to\sigma^2N^2$; see Table~\ref{tab:guta-constants} for known values of $\sigma_R$ and $\eta_R$.
Correspondingly, equating the right-hand side of equation to $\delta$, we have $\epsilon=\frac 6{\eta_R} \sqrt u (\sqrt u+\sqrt{u+1})$ with $u=\eta_R^2 \log(2d/\delta)/(18 N)$. The values used for the performance comparisons in the main text are based on Eq.~\eqref{eq:guta-general-nt}. Hence the results differ from a similar analysis in Ref.~\cite{almeida2023comparison}, where the original bound from Ref.~\cite{guta2020fast} is used which employs $\tilde w_R$ and lacks the overall factor of $2$.

\begin{table}
    \centering
    \begin{tabular}{lcccc}
    \toprule
         {Measurement scenario \hspace{4em}} & Parameters & ~ ~ $\sigma_R$ ~ ~ & ~ ~ $\eta_R$ ~ ~  \\
    \midrule
        structured measurement, dimension $d$ & $r=1$, $q=1$ & $\sqrt{2d}$ & $\sqrt{d/2}$  \\[.5ex]
        Pauli-bases, $q$ qubits & $d=2$, $r=1$ & $3^{q/2}$ & $2\,(4/3)^{q/2}$  \\[.5ex]
        Pauli observables, $Q$ qubits & $d=2^Q$, $r=d/2$, $q=1$ & $d$ & 2 \\
    \bottomrule
    \end{tabular}
    \caption{Constants $\sigma_R$ and $\eta_R$ for the confidence region $C_R$ defined via Eq.~\eqref{eq:guta-general-nt}, for different measurement scenarios. The constants are obtained from the values for $R$ and $\sigma^2$ in Appendix~C of Ref.~\cite{guta2020fast}.
    The column ``Parameters'' shows how the measurement scenario emerges from the local symmetric measurements discussed in Section~\ref{ap:sic-v3} by fixing some of the parameters $d$, $r$, and $q$.}
    \label{tab:guta-constants}
\end{table}

\subsection{Usage example: Entanglement certification}%
\label{ap:gme}

If quantum state tomography is experimentally affordable and a reasonably small confidence region can be achieved, then this region provides easy access to many properties of the state. Here we consider entanglement characterisation \cite{guehne2009entanglement}.
To verify that an experimental state is entangled using tomographic data, one must not only verify that the estimate $\hat\rho(\vec f)$ is entangled, but also that its associated confidence region contains only entangled states.
A more common approach to entanglement detection is to use entanglement witnesses \cite{terhal2000bell}, but that heavily relies on prior information on the state. When the state is randomly sampled, effectively detecting its entanglement requires a total number of $e^{\Omega(d)}$ witnesses to be tested, which is impractical \cite{znidaric2007detecting,liu2022fundamental}. In such cases, tomography is a viable approach.

\begin{figure}
    \centering
    \includegraphics[width=.5\columnwidth]{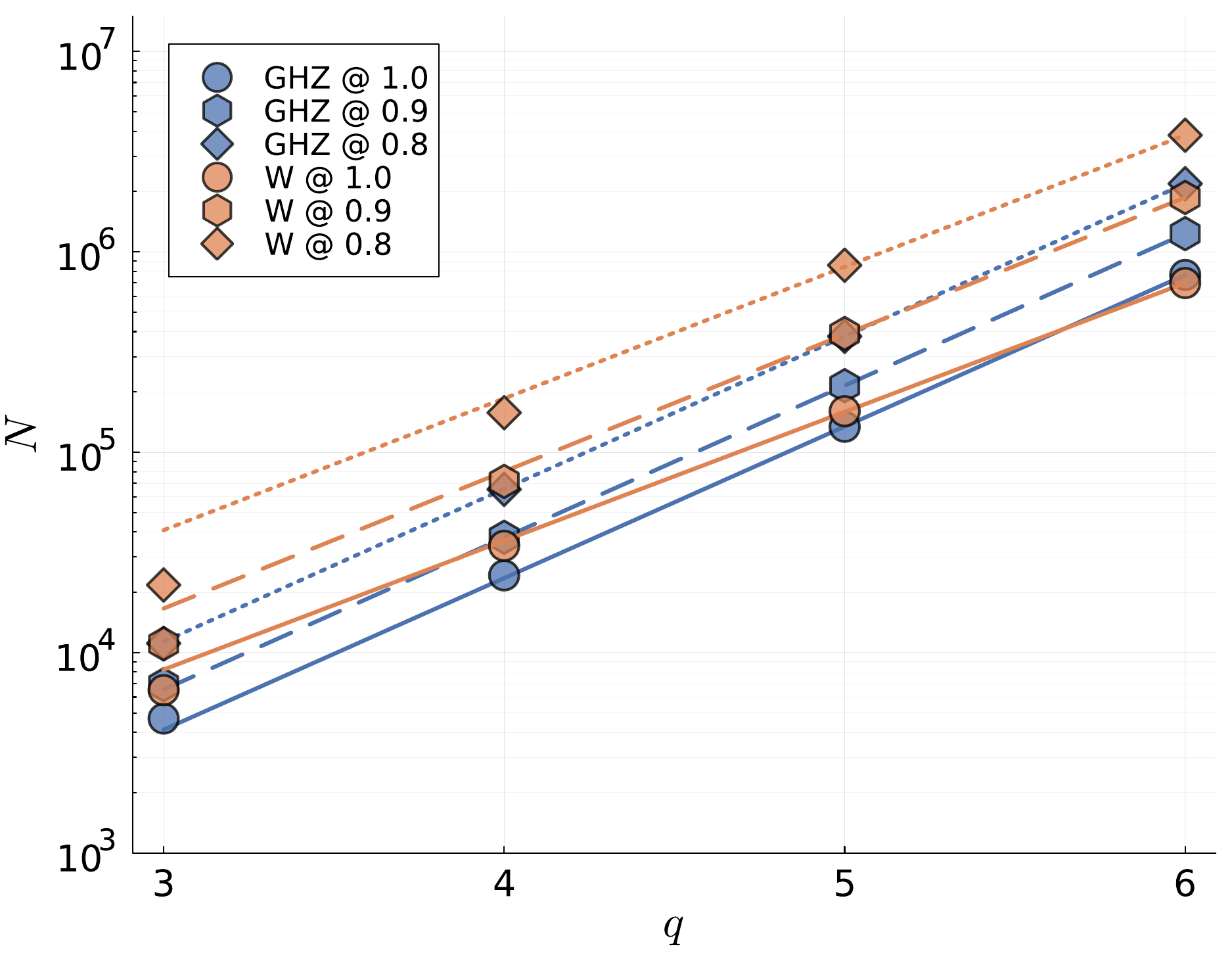}
    \caption{Number of samples $N$ required to certify genuine multipartite entanglement by means of the confidence region $C_B$ for confidence level $1-\delta=0.9$.
    We simulate tomography experiments using local symmetric informationally-complete  measurements and determine $N$ over $1024$ estimation procedures, such that in $(50\pm 2)\%$ of the cases the confidence region contains no state which is a PPT mixture. Results are shown for $\mathcal D_t(\ket{\mathrm{GHZ}})$ and $\mathcal D_t(\ket{\mathrm{W}})$ with $t=1.0$, $0.9$, $0.8$ and $\mathcal{D}_t(\ket\psi) = t\dyad{\psi} + (1-t)\eye/2^q$.
    The lines correspond to least-squares fits assuming an exponential model. For reference, the largest values for $t$ such that a GHZ (W) state is a PPT mixture are $0.429$, $0.467$, $0.484$, $0.493$ ($0.479$, $0.474$, $0.422$, $0.353$) for $q = 3$, $4$, $5$, $6$ qubits, respectively.}
    \label{fig:gme}
\end{figure}

In multipartite systems, the strongest notion of entanglement is that of genuine multipartite entanglement \cite{seevinck2008partial,guehne2009entanglement}, which is a property of any state that cannot be written as a convex combination of biseparable states.
More precisely, let $\rho_{ABC}$ be a three-partite state with the subsystems labelled by $A$, $B$ and $C$. Then $\rho$ is said to be genuine multipartite entangled if and only if \cite{guehne2009entanglement, horodecki_entanglement,seevinck2008partial}
\begin{equation}\label{eq:bisep}
    \rho_{ABC} \neq p_1 \rho_{A \vert BC} + p_2 \rho_{AB \vert C} + p_3 \rho_{AC \vert B},
\end{equation}
where $(p_1,p_2,p_3)$ forms a probability distribution and $\rho_{X \vert Y}$ stands for a state which is separable across the bipartition $X|Y$, that is, $\rho = \sum_i p_i \rho_{X}^i \otimes \rho_{Y}^i$. This definition extends naturally for $q$-partite systems, in which case all $\binom{q}{2}$ bipartitions must be considered.
The use of biseparability in this definition can be relaxed to that of a bipartite state with positive partial-transpose (PPT), leading to the PPT mixtures \cite{jungnitsch2011taming}. Although the set of PPT mixtures is a superset of the biseparable states, it is frequently a good approximation to the former \cite{jungnitsch2011taming}. Most importantly, it has a finite semidefinite representation given by
\begin{equation}\label{eq:gme-sdp}
    \Set{ \rho = \sum_i \omega_i | \omega_i \succcurlyeq 0 \;\forall i,\quad \omega_i^{\intercal_i} \succcurlyeq 0\;\forall i,\quad \sum_i \tr(\omega_i) = 1 }
\end{equation}
Here, $i$ runs through all possible bipartitions of $q$ qubits and by $\rho_i^{\intercal_i}$ we mean a partial transposition with respect to the cut imposed by the $i$th bipartition. The constraints enforce that $\rho = \sum_i \omega_i$ is a PPT mixture, as one can interpret each $\omega_i$ as a subnormalised PPT state \cite{jungnitsch2011taming,tavakoli2023semidefinite}.

Our goal is to certify that a confidence region contains only genuinely entangled states and we use region $C_B$ as an example. Since it can be represented by the conic constraint $\norm*{\rho - \hat{\rho}}_M \leq \epsilon \sigma_B$, it can be directly added to Eq.~\eqref{eq:gme-sdp}. Then, a feasible solution would certify there exists a state $\rho$ which is not genuine multipartite entangled inside the confidence region. In this case, we cannot guarantee with high confidence that the true state $\rho$ is genuine multipartite entangled, see also Section \ref{sec:usage} in the main text.

As a demonstration, we consider $q= 3, \ldots, 6$ qubits, each measured with local symmetric informationally complete measurements \cite{renes2004sic}.
This scheme is experimentally feasible \cite{stricker2022experimental} and has the same performance as the Pauli-bases measurement with respect to $C_A$ and $C_B$ (see Section~\ref{ap:sic-v3}). Notice that the constants necessary to evaluate the reference confidence region $C_R$ are not available for this measurement scheme.
With these measurements, we simulate a tomography experiment using Greenberger--Horne--Zeilinger states $\ket{\mathrm{GHZ}}=\frac1{\sqrt2}(\ket 0^{\otimes q}+ \ket 1^{\otimes q})$ and W states $\ket{\mathrm W}=\frac1{\sqrt d}(\ket{00\dotsm001}+\ket{00\dotsm010}+\dotsm+\ket{10\dotsm000})$ as the fiducial states. The number of samples $N$ is then increased until the associated confidence region stops intersecting the set of PPT mixtures. The results are shown in Figure~\ref{fig:gme}. As expected, we observe an exponential increase of the sample cost $N$ with the number of qubits, and generally more samples are required for noisier states. Overall, the number of samples is still moderate, with millions of state preparations required for 6 qubits.

\end{appendix}

\end{document}